\documentclass[11pt,reqno]{amsart}
\usepackage[T1]{fontenc}
\usepackage[colorlinks=true, pdfstartview=FitV, linkcolor=blue, citecolor=blue, urlcolor=blue]{hyperref}
\usepackage{graphicx}
\usepackage{amssymb,amsmath,mathrsfs,amsthm,bm,braket,marginnote}
\usepackage{enumerate}
\usepackage{tikz}
\usepackage{multirow}
\numberwithin{equation}{section}
\newtheorem{theorem}{Theorem}[section]
\newtheorem{lemma}[theorem]{Lemma}
\newtheorem{define}[theorem]{Definition}
\newtheorem{remark}[theorem]{Remark}

\newtheorem{example}[theorem]{Example}

\newtheorem{proposition}[theorem]{Proposition}

\reversemarginpar

\def\J{\mathrm{J}}

\def\D{\mathcal{D}}

\def\DO{\mathrm{DO}}
\def\ker{\mathrm{Ker\,}}
\def\Ad{\mathrm{Ad\,}}
\def\PBDO{\mathrm{PBDO}}
\def\SDO{\mathrm{SDO}}
\def\hhh{{H\tilde{\times} H\tilde{\times}H}}

\author[B. Wang]{Bao Wang}
\address{School of Mathematics and Statistics, Ningbo University, Ningbo 315211, PR China.}
\email{wangbao@nbu.edu.cn}

\begin{document}

\title[Generalized pentagram maps via $Q$-nets and refactorization mappings]{Generalized pentagram maps via $Q$-nets and refactorization mappings}


\maketitle

\begin{abstract}
We introduce a family of generalizations of the pentagram maps related to $Q$-nets.
A specific example is considered,
and we find 
the map can be treated as a refactorization mapping in the Poisson-Lie group of pseudo-difference operators.
This method was firstly proposed by Izosimov,
and we generalize it to fit our needs.
Using this description,
we obtain the corresponding Lax form with a spectral parameter and invariant Poisson brackets.
Finally,
we consider the reduction to $B$-nets and the discrete BKP equation,
offering a geometric explanation for the discrete-time Toda equation of BKP type proposed by Hirota.

\end{abstract}

\section{Introduction}

The pentagram map was introduced by Schwartz \cite{Schwartz1992The} in 1992,
and now it has been a popular discrete dynamic system.
It is defined on the space of plane polygons in projective plane and has
an elegant geometric meaning:
taking a vertex $v_k$ to the intersection of two segments $\overline{v_{k-1}v_{k+1}}$ and $\overline{v_kv_{k+2}}$.
The pentagram map has many amazing connections to integrable systems \cite{2010The,2013Liouville,Soloviev2013Integrability}, 
cluster algebras \cite{2014Integrable,GLICK20111019,cluster2024}, 
dimer models \cite{dimer2022}, Poisson-Lie group \cite{refactorization2020}, Poncelet polygons \cite{izosimov2022pentagram} etc.
In particular,
its continuous limit is the Boussinesq equation \cite{2010The}.

There are some generalizations of the pentagram map,
such as 
the high dimensional map on so-called corrugated polygons \cite{2014Integrable},
the short-diagonal pentagram map in $\mathbb{RP}^d$ \cite{2012Integrability},
the dented pentagram map \cite{dented},
generalized pentagram maps from Y-meshes \cite{glick2016},
the long-diagonal pentagram map \cite{long2022},
the pentagram map on coupled polygons \cite{coupled_sub},
the maps on Grassmannian \cite{grassmann,Non} 
and non-integrable pentagram-type maps \cite{nonintegrable2015}.
\subsection{$Q$-nets and generalized pentagram maps}
In this paper,
we introduce a family of generalization of pentagram map related to $Q$-nets.
Discrete conjugate nets (or multidimensional quadrilateral lattice, or $Q$-nets) were introduced by Doliwa-Santini \cite{doliwaMQL},
providing an interesting description of the relation between integrability and discrete geometry.
A systematic exposition for that is given in the monograph of Bobenko-Suris \cite{bobenko2008discrete}.
It is an open question of how the pentagram map and $Q$-nets relate \cite[Remark 1.5]{glick2016}.
Recently,
Affolter, Glick, Pylyavskyy and Ramassamy \cite{affolter2024vector} resolve the question by studying a geometric model for local transformation of bipartite graphs,
whose evolution for different choices of the graph coincides with many notable dynamical systems including the pentagram map, $Q$-nets and discrete Darboux maps.
In other words,
they have placed the pentagram map and $Q$-nets into a unified framework.

Here
we try a different approach:
consider the reduction of $Q$-nets,
which can be treated as generalized pentagram maps.
Let's
consider the following example.
Suppose $v^1=\{v_i^1\}_{i\in\mathbb{Z}}$, $v^2=\{v_i^2\}_{i\in\mathbb{Z}}$, are a pair of twisted $N$-gons in $\mathbb{RP}^3$ with the same monodromy.
Define a map $T_Q$
\begin{align}\label{map1}
	T_Q\left(v_i^{1},v_i^{2}\right)=\left(v_i^{2},v_i^{3}\right),\quad
	\forall i\in\mathbb{Z},
\end{align}
where $v_i^3$ is the intersection of three planes in $\mathbb{RP}^3$
\begin{align}
	\begin{split}
		v_i^{3}:=\left(v_{i-1}^{1},\,v_{i-1}^{2},\,v_{i}^{2}\right)
		\cap
		\left(v_{i}^{1},\,v_{i-1}^{2},\,v_{i+1}^{2}\right)
		\cap
		\left(v_{i+1}^{1},\,v_{i}^{2},\,v_{i+1}^{2}\right).
	\end{split}
\end{align}
It is clear that three sets of four vertices $\left(v_{i-1}^{1},\,v_{i-1}^{2},\,v_{i}^{2},\,v_i^3\right)$,
$\left(v_{i}^{1},\,v_{i-1}^{2},\,v_{i+1}^{2},\,v_i^3\right)$
and
$\left(v_{i+1}^{1},\,v_{i}^{2},\,v_{i+1}^{2},\,v_i^3\right)$
are coplanar respectively (see Figure \ref{f0}).
\begin{figure}[htbp]
	\centering
	\includegraphics[scale=0.6]{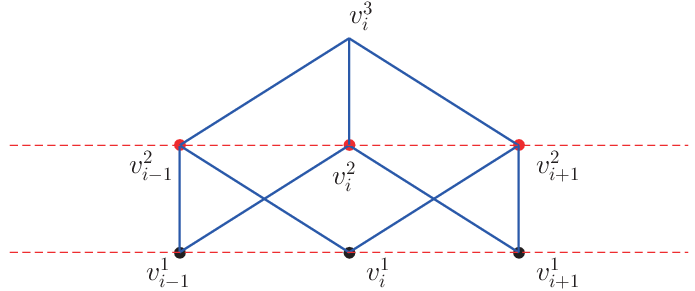}
	\caption{An example of $Q$-type pentagram maps}
	\label{f0}
\end{figure}

\subsection{Refactorization mappings and $Q$-type pentagram maps}\label{subsec_re}
The main result is that
we find the map \eqref{map1} can be treated as a refactorization mapping in the Poisson-Lie group of pseudo-difference operators.
The approach is proposed by Izosimov in \cite{refactorization2020},
which provides new Lax forms with a spectral parameter and invariant Poisson structures for various multidimensional pentagram maps.

Suppose that we are given three $2N$-periodic difference operators $\D_i$, $i=1,2,3$,
such that
$\D_1$ is supported in $\{2,3\}$, $\D_2$ is supported in $\{0,1\}$,
and $\D_3=f+gT+hT^2$ is supported in $\{0,1,2\}$ with $f_{2i}=h_{2i}=0$, for all $i\in\mathbb{Z}$.
Then we can assign a twisted $2N$-gon as follows:
a sequence of vectors $\Phi_i\in\mathbb{R}^4$, $i\in\mathbb{Z}$ solves
\begin{align*}
\left(\D_1\D_3-\D_2\right)\Phi=0.
\end{align*} 
A key ansatz is that we can construct a pair of twisted $N$-gons $(v^1,v^2)\subset\mathbb{RP}^3$ with lifts $(V^1,V^2)\subset\mathbb{R}^4$ as follows:
\begin{align*}
V_i^1=\Phi_{2i},\quad
V_i^2=\Phi_{2i+1},\quad
\forall i\in\mathbb{Z}.
\end{align*}
It is clear that if we transform our operators as
\begin{align*}
	\D_1\mapsto \alpha\D_1\beta^{-1},\quad
	\D_2\mapsto\alpha\D_2\gamma^{-1},\quad
	\D_3\mapsto\beta\D_3\gamma^{-1},
\end{align*}
where $\alpha$, $\beta$, $\gamma          $ are quasi-periodic sequences with the same monodromy,
then
the pair of polygons $(v^1,v^2)$ will remain unchanged.
First, we prove the following theorem.
\begin{theorem}
There is a bijection between triples of operators $\D_i$, $i=1,2,3$, up to the above action,
and pairs of twisted $N$-gons $(v^1,v^2)$ modulo projective transformations.
\end{theorem}

Then we consider the following refactorization mapping $\D_i\mapsto \tilde{\D}_i$, $i=1,2,3$, determined by:
\begin{align}\label{eq_i1}
	\left(\tilde{\D}_1^{-1}\tilde{\D}_2\right)\D_3=\tilde{\D}_3\left(\D_1^{-1}\D_2\right).
\end{align}
Recall that the classic refactorization mapping in \cite{refactorization2020} is
\begin{align*}
\tilde{\D}_+\D_-=\tilde{\D}_-\D_+,
\end{align*}
where $\D_\pm$ are difference operators.
However,
our map \eqref{eq_i1} involves a pseudo-difference operator $\D_1^{-1}\D_2$ and a difference operator $\D_3$.
In this case,
we prove the map \eqref{eq_i1} is still well-defined.
\begin{theorem}
The equation \eqref{eq_i1} gives a well-defined mapping of the space of triples of operators, up to the action \eqref{map1}, to itself.
\end{theorem}
Define
\begin{align*}
	\mathcal{L}:=\D_3^{-1}\left(\D_1^{-1}\D_2\right),\quad
	\tilde{\mathcal{L}}:=\tilde{\D}_3^{-1}\left(\tilde{\D}_1^{-1}\tilde{\D}_2\right),
\end{align*}
then the equation \eqref{eq_i1} gives a Lax representation 
\begin{align*}
\mathcal{L}\mapsto \tilde{\mathcal{L}}=\D_3\mathcal{L}\D_3^{-1}=\left(\D_1^{-1}\D_2\right)\D_3^{-1}.
\end{align*}
Finally, we provide the following connection.
\begin{theorem}
The $Q$-type pentagram map \eqref{map1}, written in terms of difference operators,
is just the refactorization mapping determined by \eqref{eq_i1}.
\end{theorem}
\subsection{Reduction to the discrete BKP equation}
The author previously considered the relation of pentagram-type maps and the discrete KP equation in
\cite{wang2023pentagram},
posing an open question on how to relate pentagram-type maps and the multi-component discrete KP equation.
It is known that the $Q$-nets are related to the
multi-component discrete KP equation \cite{doliwa1999generating},
so this paper resolves that question.
Finally,
we consider the reduction to the discrete BKP equation
\begin{align}
	\tau\tau_{123}=\tau_{12}\tau_3-\tau_{13}\tau_2+\tau_{23}\tau_1,
\end{align}
whose  linear problem
was constructed in \cite{Nimmo1997}.
The equation can also be obtained from some reduction of the quadrilateral lattice,
which was called $B$-quadrilateral lattice (BQL) \cite{Doliwa_BQL}.
We will consider the corresponding reductions for the reduced $Q$-nets and the $Q$-type pentagram maps.
The latter turns out to be the discrete Toda equation of BKP type proposed by Hirota \cite{hirota_iwao_tsujimoto_2001}.

The reminder of this paper is organized as follows.
Section \ref{sec_back} covers the background on difference operators and refactorization mappings.
In Section \ref{sec_qnet},
we illustrate how to define reduced $Q$-nets and $Q$-type pentagram maps.
The cross-ratios coordinates and explicit formulas for the map \eqref{map1} are demonstrated in Section \ref{sec_expl}.
The interpretation of refactorization mappings for $Q$-type pentagram maps is given in Section \ref{sec_refa}.
Section \ref{sec_redu} discusses the reduction of $Q$-type pentagram maps to B-type discrete equations.
Finally,
we conclude with some indications of future work.
\section{A brief review on difference operators and refactorization mappings}\label{sec_back}
\subsection{Difference and pseudo-difference operators}
In this subsection,
we recall some basic notions and facts about (pseudo-)difference operators  \cite{refactorization2020,izosimov2022pentagram}.

Let $\mathbb{R}^\infty$ be the vector space of bi-infinite sequences of real numbers,
and let $\J$ be a finite collection of integers.
A linear operator $\D:\mathbb{R}^\infty\rightarrow \mathbb{R}^\infty$ is called a difference operator supported in $\J$ if it can be written as
\begin{align*}
	(\D \xi)_k=\sum_{j\in\J}a_{j,k}\xi_{k+j},
\end{align*}
or equivalently,
\begin{align*}
	\D=\sum_{j\in\J}a_jT^j,
\end{align*}
where $T:\mathbb{R}^\infty\rightarrow\mathbb{R}^\infty$ is the shift operator $(T\xi)_k=\xi_{k+1}$.
Here $a_j$ are bi-infinite sequences in $\mathbb{R}^\infty$.

The order of a difference operator is defined as the number $\max\J-\min\J$.
A difference operator is called properly bounded if none of the sequences $a_{\min \J}$, $a_{\max \J}$ vanish.
A difference operator $\D$ is called periodic if all its coefficients $a_j$ are $n$-periodic sequences.
The set of $n$-periodic difference operators supported in $\J$ is denoted by $\DO_n(\J)$,
while $\PBDO_n(\J)\subset\DO_n(\J)$ stands for the (dense) subset of properly bounded operators. 

Furthermore,
let $H\tilde{\times}H$ be the group that consists of pairs of non-vanishing $n$-quasi-periodic scalar sequences with the same monodromy
\begin{align*}
	H\tilde{\times}H:=\{(\lambda,\mu)\mid \forall i\in\mathbb{Z},\,\lambda_i\mu_i\neq 0,\,\text{and}\,\,\exists z\in\mathbb{R}^*\,\text{s.t.}\,\forall i\in\mathbb{Z},\,\lambda_{i+n}=z\lambda_i,\,\mu_{i+n}=z\mu_i\},
\end{align*}
which acts on the space $\DO_n(\J)$ by means of the action
\begin{align}\label{eq_hh}
	\D\rightarrow \lambda\D\mu^{-1}.
\end{align}
Then we have the following one-to-one correspondence \cite{refactorization2020}.
\begin{proposition}\label{pro_equi0}
	The polygons in $\mathbb{RP}^d$ modulo the projective transformations can be encoded by means of difference operators in $\PBDO_n(\J)/H\tilde{\times}H$ with $\J=\{0,1,\ldots,d+1\}$, that is
	\begin{align}\label{op_D}
		\D=a^0+a^1T+a^2T^2+\cdots+a^{d+1}T^{d+1},
	\end{align}
	where $a^j\in\mathbb{R}^\infty$ are sequences such that $a^0_k\neq 0$, $a^{d+1}_k\neq 0$ for any $k\in\mathbb{R}$. 
\end{proposition}

The $n$-periodic pseudo-difference operator is a formal Laurent series of $T$
\begin{align}\label{eq_a}
	\sum_{j=k}^{+\infty}a_jT^j,
\end{align}
where each $a_j$ is an $n$-periodic bi-infinite sequence.
The set of $n$-periodic pseudo-difference operators, denoted by $\Psi\DO_n$,
is an associative algebra with respect to addition and composition of operators.
The operator \eqref{eq_a} is invertible if none of the elements of $a_k$ vanish.
The set of all invertible operators, denoted by $\mathrm{I}\Psi\DO_n$, is a group with respect to composition of operators,
which can be regarded as an infinite-dimensional Lie group with Lie algebra $\Psi\DO_n$.

Recall that a Lie group $G$ equipped with a Poisson structure is called a Poisson-Lie group if the group multiplication $G\times G\rightarrow G$ is a Poisson map.
Note that $\mathrm{I}\Psi\DO_n$ is isomorphic to the group of matrices over Laurent series with non-vanishing determinant,
which is a version of the loop group of $GL_n$.
The Poisson structure on the latter is known,
so
there also exists a natural Poisson structure on the group $\mathrm{I}\Psi\DO_n$ \cite{refactorization2020}.
\subsection{The refactorization maps in Poisson-Lie groups}
The classical pentagram map can be treated as a refactorization map in the Poisson-Lie group of pseudo-difference operators \cite{refactorization2020}.
Below we briefly describe the construction.

For a twisted $n$-gon $\{v_k\}_{k\in\mathbb{Z}}\in\mathbb{RP}^2$,
its lift $\{V_k\}_{k\in\mathbb{Z}}$ in $\mathbb{R}^3$ satisfies a recurrence relation 
\begin{align*}
	a_kV_{k+3}+b_kV_{k+2}+c_kV_{k+1}+d_kV_k=0.
\end{align*}
It can be equivalently written as $\mathcal{D}V=0$ with $\D=aT^3+bT^2+cT+d$.
According to proposition \ref{pro_equi0},
there is a one-to-one correspondence between the $n$-periodic difference operators (up to the action of scalar sequences) and the twisted $n$-gons (up to projective transformations).

If we split the operator into two parts $\mathcal{D}=\mathcal{D}_++\mathcal{D}_-$,
where $\mathcal{D}_\pm$ are supported in $\J_\pm$ with $\J_+=\{0,2\}$, $\J_-=\{1,3\}$,
then the pentagram map is equivalent to a multivalued map
\begin{align*}
	\mathcal{D}=\mathcal{D}_++\mathcal{D}_-\longmapsto \tilde{\mathcal{D}}=\tilde{\mathcal{D}}_++\tilde{\mathcal{D}}_-
\end{align*}
determined by the condition $\tilde{\mathcal{D}}_+\mathcal{D}_-=\tilde{\mathcal{D}}_-\mathcal{D}_+$,
or equilentely, $\tilde{\mathcal{D}}^{-1}_-\tilde{\mathcal{D}}_+=\mathcal{D}_+\mathcal{D}_-^{-1}$.
Define $\mathcal{L}:=\mathcal{D}_-^{-1}\mathcal{D}_+$,
then the map $\mathcal{L}\rightarrow\tilde{\mathcal{L}}:=\mathcal{D}_+\mathcal{D}_-^{-1}$ is a refactorization map in the Poisson-Lie group of pseudo-difference operators.

Various integrable pentagram-type maps have such a similar description \cite{refactorization2020} (see some examples in Table \ref{tab1}).
\begin{table}[htbp]
	\caption{Examples of pentagram-type maps corresponding to disjoint progressions}
	\label{tab1}
	\begin{tabular}{c|c|c}
		\hline
		$\mathrm{J}_+$ & $\mathrm{J}_-$& \textbf{The corresponding map} \\
		\hline
		$\{0,2\}$&$\{1,3\}$&Classical pentagram map\\
		\hline
		$\{0,1\}$&$\{2,3\}$&Inverse pentagram map\\
		\hline
		$\{0,d\}$&$\{1,d+1\}$&Pentagram map on corrugated polygons in $\mathbb{RP}^d$\\
		\hline
		$\{0,2,\ldots,2k\}$&$\{1,3,\ldots,2k+1\}$&Short-diagonal pentagram map in $\mathbb{RP}^{2k}$\\
		\hline
		$\{0,1,\ldots,m\}$&$\{m+1,\ldots,d+1\}$&Inverse dented pentagram maps in $\mathbb{RP}^d$\\
		\hline
	\end{tabular}
\end{table}
This description automatically provides new Lax forms with a spectral parameter and invariant Poisson structures.

\section{Reduced $Q$-nets and $Q$-maps}\label{sec_qnet}

Recall that a $3$-dimensional $Q$-net in $\mathbb{R}^N$, $N\geq 3$, is defined as a map $f:\mathbb{Z}^3\rightarrow \mathbb{R}^N$ such that all the elementary quadrilaterals with vertices
\begin{align*}
f,\quad T_if,\quad
T_jf,\quad
T_{i}T_jf
\end{align*} 
are coplanar \cite{doliwaMQL,bobenko2008discrete},
where $T_i$ is the shift operator in the $i$th direction.

Firstly,
we replace $\mathbb{R}^N$ with the projective space $\mathbb{RP}^N$.
Note that the map $f$ depends on three variables in $\mathbb{Z}$.
In the following,
we consider its reduction to two discrete variables.
\subsection{Reduced $Q$-net}
\begin{define}
	Let $a$, $b$, $c\in\mathbb{Z}^2$ be distinct and assume $1\leq a_2\leq b_2\leq c_2$, $a\neq b\neq c\neq a$.
	Say that $Q=\{a,b,c\}$ is a $Q$-pin if the vectors $\overrightarrow{Oa}$, $\overrightarrow{Ob}$, $\overrightarrow{Oc}$ are not colinear in $\mathbb{Z}^2$,
	where $O$ is the origin point $(0,0)$.
	When $Q$ is a $Q$-pin,
	we will always assume its elements are called $a$, $b$, $c$ and $1\leq a_2\leq b_2\leq c_2$.
\end{define}
\begin{remark}
	If $a_2=b_2$ or $b_2=c_2$,
	then we allow any choice of $a$, $b$, $c$ satisfying $1\leq a_2\leq b_2\leq c_2$.
	The corresponding dynamical system we define will not depend on the choice.
\end{remark}

\begin{define}\label{def_Qpin}
	Let $Q=\{a,b,c\}$ be a $Q$-pin and suppose $N\geq 3$.
	A reduced $Q$-net of type $Q$ is a grid of points $v_{i,j}$ and planes $P_{i,j}^k$ ($k=1,2,3$) in $\mathbb{RP}^N$ and such that
	
	(i) $v_r$, $v_{r+a}$, $v_{r+b}$, $v_{r+a+b}$ all lie on $P_r^1$ and are distinct for all $r\in\mathbb{Z}^2$;
	
	(ii) $v_r$, $v_{r+b}$, $v_{r+c}$, $v_{r+b+c}$ all lie on $P_r^2$ and are distinct for all $r\in\mathbb{Z}^2$;
	
	(iii) $v_r$, $v_{r+c}$, $v_{r+a}$, $v_{r+c+a}$ all lie on $P_r^3$ and are distinct for all $r\in\mathbb{Z}^2$;
	
	(iv) Any three points in $v_r$, $v_{r+i}$, $v_{r+j}$, $v_{r+i+j}$ are not colinear, where $i$, $j\in\{a,b,c\}$, $i\neq j$ and $r\in\mathbb{Z}^2$;
	
	(v) $P_r^1$, $P_r^2$, $P_r^3$ are distinct for all $r\in\mathbb{Z}^2$;
	
	(vi) $P_{r+c}^1$, $P_{r+a}^2$, $P_{r+b}^3$ are distinct for all $r\in\mathbb{Z}^2$.
\end{define}

\begin{remark}
	Two $Q$-pins $Q$ and $Q'$ are equivalent if $Q'=g(Q)$ for some map $g:\mathbb{Z}^2\rightarrow\mathbb{Z}^2$ of the form 
	\begin{align*}
		g(i,j)=(\pm i+f(j),j),
	\end{align*}
	with $f:\mathbb{Z}\rightarrow \mathbb{Z} $. 
	In this case,
	$g$ sends rows to rows and preserve the $j$ direction,
	so $Q$ and $Q'$ are essentially the same objects.
	When $(i,j)\rightarrow(i,-j)$,
	then $Q'$ is the inverse of $Q$.
\end{remark}

The reduced $Q$-net can be treated as a geometric dynamical system.
Let $\mathcal{U}_N$ be the space of infinite polygons $V$ with vertices $\{v_i\}_{i\in\mathbb{Z}}$ in $\mathbb{RP}^N$,
considered modulo projective equivalence.
Let $\mathcal{U}_{N,m}$ denote the space of $m$-tuples $\{V^{(1)},\ldots, V^{(m)}\}$.
Let $V^{(j)}$ denote the $j$th row of a grid $\{v_{i,j}\}_{i,j\in\mathbb{Z}}$,
that is $V_i^{(j)}=v_{i,j}$.
We begin with the case of $\mathbb{RP}^3$,
then generalize it to higher dimensions.

\subsection{The $Q$-maps in $\mathbb{RP}^3$}

\begin{proposition}
	Let $Q=\{a,b,c\}$ be a $Q$-pin and let $P$ be a reduced $Q$-net of type $Q$.
	Then
	\begin{align}\label{def_1}
		v_{r+a+b+c}
		=\left<v_{r+a},v_{r+a+b},v_{r+a+c}\right>
		\cap
		\left<v_{r+b},v_{r+a+b},v_{r+b+c}\right>
		\cap
		\left<v_{r+c},v_{r+a+c},v_{r+b+c}\right>
	\end{align}
	and
	\begin{align}\label{def_2}
		v_{r}
		=\left<v_{r+a},v_{r+b},v_{r+a+b}\right>
		\cap
		\left<v_{r+b},v_{r+c},v_{r+b+c}\right>
		\cap
		\left<v_{r+a},v_{r+c},v_{r+a+c}\right>
	\end{align}
	for all $r\in\mathbb{Z}^2$.
\end{proposition}
\begin{proof}
	By definition of a reduced $Q$-net,
	the points $v_{r},$ $v_{r+a}$, $v_{r+b}$, $v_{r+a+b}$ are distinct and both lie on the plane $P_r^1$.
	So we have $v_r\in P_r^1=\left<v_{r+a},v_{r+b},v_{r+a+b}\right>$.
	Also,
	we have $v_r\in P_r^2=\left<v_{r+b},v_{r+c},v_{r+b+c}\right>$
	and 
	$v_r\in P_r^3=\left<v_{r+a},v_{r+c},v_{r+a+c}\right>$.
	Finally,
	$P_r^1$, $P_r^2$ and $P_r^3$ are distinct and both contain $v_r$,
	proving \eqref{def_2}.
	The proof of \eqref{def_1} is similar.
\end{proof}

In \eqref{def_1},
$v_{r+a+b+c}$ has the highest $j$-value and $v_{r+a}$ has the smallest $j$-value.
In other words,
all these points in \eqref{def_1} lie within rows from $j=r_2+a_2$ to $j=r_2+a_2+b_2+c_2$.
As such,
\eqref{def_1} can be seen as a recurrence that inputs $b_2+c_2$ rows of points and determine the next row.
Similarly, in \eqref{def_2},
$v_{r+b+c}$ has the highest $j$-value and $v_r$ has the smallest $j$-value.
These points lie within rows from $j=r_2$ to $j=r_2+b_2+c_2$.
So \eqref{def_2}  can be used to determine the previous row using $b_2+c_2$ rows of points.

Let $m=b_2+c_2$,
and suppose $V^j=(V^{(j+1)},\ldots,V^{(j+m)})=(v_{\cdot,j+1},\ldots,v_{\cdot,j+m})\in\mathcal{U}_{3,m}$ with $j\in\mathbb{Z}$.
We impose the relations on $V^j$ of two types, 
defined as follows.

\vspace{10pt}
(P1) Points $v_{r}$, $v_{r+a}$, $v_{r+b}$, $v_{r+a+b}$ are coplanar for all $r$ with $j+1\leq r_2\leq j+c_2-a_2$;

(P2) Points $v_{r}$, $v_{r+a}$, $v_{r+c}$, $v_{r+a+c}$ are coplanar for all $r$ with $j+1\leq r_2\leq j+b_2-a_2$.

\vspace{10pt}
Note that a relation like $v_r$, $v_{r+b}$, $v_{r+c}$, $v_{r+b+c}$ coplanar will never fit within $m$ consecutive rows.
Define $X_{3,Q}\subset\mathcal{U}_{3,m}$ by all elements in $\mathcal{U}_{3,m}$ which satisfy $(P1)$ and $(P2)$ relations.
Define $V^{(j)}\in\mathcal{U}_{3}$ using \eqref{def_2} and define $V^{(j+m+1)}\in\mathcal{U}_{3}$ using \eqref{def_1}.
Then we can write
\begin{align*}
	F(V^j)=(V^{(j+2)},\ldots,V^{(j+m+1)}),\quad
	G(V^j)=(V^{(j)},\ldots,V^{(j+m-1)}).
\end{align*}

\begin{proposition}\label{pro_3}
	Generically,
	$F(V^j)$, $G(V^j)\in X_{3,Q}$ and they are inverse to each other.
\end{proposition}
\begin{proof}
	First, we show $F(V^j)\in X_{3,Q}$.
	Since $V^j=(V^{(j+1)},\ldots,V^{(j+m)})\in X_{3,Q}$,
	it suffice to verify $(P1)$ and $(P2)$ conditions involving $V^{(j+m+1)}$.
	For $(P1)$,
	we should prove that $v_{r}$, $v_{r+a}$, $v_{r+b}$, $v_{r+a+b}$ are coplanar for $r_2=j+c_2-a_2+1$.
	The formula \eqref{def_1} with $r_2+a_2+b_2+c_2=j+m+1$, that is $r_2=j-a_2+1$, 
	defines $V^{(j+m+1)}$, confirming $(P1)$.
	For $(P2)$,
	we should prove that the points $v_{r}$, $v_{r+a}$, $v_{r+c}$, $v_{r+a+c}$ are coplanar for $r_2=j+b_2-a_2+1$.
	It can also be seen from \eqref{def_1} by taking $r_2=j-a_2+1$.
	A similar argument shows $G(V^j)\in X_{3,Q}$.
	
	Let $V^j=(V^{(j+1)},\ldots,V^{(j+m)})\in X_{3,Q}$ and use $F$ to build $V^{(j+m+1)}$.
	Showing $$G(F(V^j))=V^j$$ amounts to verifying that \eqref{def_2} holds for $r$ with $r_2=j+1$.
	According to the first term of the right hand side of \eqref{def_1},
	$v_{r+a}$, $v_{r+a+b}$, $v_{r+a+c}$, and $v_{r+a+b+c}$ are coplanar.
	This relation shifted by $-a$ shows that $v_r\in\left<v_{r+b},v_{r+c},v_{r+b+c} \right>$,
	where we take $r_2=j+1$.
	Similarly,
	if $c_2=b_2=a_2$,
	according to the second and third terms on the right hand side of \eqref{def_1},
	we have
	\begin{align}\label{eq_ks}
	v_r\in\left<v_{r+a},v_{r+c},v_{r+a+c} \right>,\quad
	v_r\in\left<v_{r+a},v_{r+b},v_{r+a+b} \right>\quad\text{with} \quad r_2=j+1.
	\end{align}
	If $a_2<b_2$ or $a_2<c_2$,
	the conditions \eqref{eq_ks} can be obtained directly from  $(P1)$ or $(P2)$.
	A similar argument shows $F(G(V^j))=V^j$.
\end{proof}
\begin{remark}
According to the proposition above,
a reduced $Q$-net can be determined by the initial points of $b_2+c_2$ consecutive rows satisfying (P1) and (P2) through the iterations of \eqref{def_1} and \eqref{def_2}.
\end{remark}
\begin{remark}
	For a $Q$-pin with $a_2=b_2=c_2$,
	the conditions $(P1)$ and $(P2)$ don't give any constraints,
	so we have $X_{3,Q}=\mathcal{U}_{3,m}$.
\end{remark}
\begin{example}
	Let $Q=\{(a_1,1),(b_1,1),(c_1,1)\}$ with $a_1,b_1,c_1\in\mathbb{Z}$, 
	which is a $Q$-pin.
	It is obvious that $m=1+1=2$,
	so we can take $(V^{(1)},V^{(2)})\in\mathcal{U}_{3,2}$ as initial polygons.
	Since $a_2=b_2=c_2=1$,
	we have $X_{3,Q}=\mathcal{U}_{3,2}$.
	The corresponding map is defined as
	\begin{align}
		\begin{split}
			v_{i+a_1+b_1+c_1,3}
			&=\left(v_{i+a_1,1},\,v_{i+a_1+b_1,2},\,v_{i+a_1+c_1,2}\right)\\
			&\cap
			\left(v_{i+b_1,1},\,v_{i+a_1+b_1,2},\,v_{i+b_1+c_1,2}\right)\\
			&\cap
			\left(v_{i+c_1,1},\,v_{i+a_1+c_1,2},\,v_{i+b_1+c_1,2}\right),
		\end{split}
	\end{align}
	while the inverse map is
	\begin{align}
		\begin{split}
			v_{i,0}
			=\left(v_{i+a_1,1},\,v_{i+b_1,1},\,v_{i+a_1+b_1,2}\right)&\cap
			\left(v_{i+b_1,1},\,v_{i+c_1,1},\,v_{i+b_1+c_1,2}\right)\\
			&\cap\left(v_{i+c_1,1},\,v_{i+a_1,1},\,v_{i+a_1+c_1,2}\right),
		\end{split}
	\end{align}
	for all $i\in\mathbb{Z}$.
	
	In particular,
	taking $a_1=0$, $b_1=1$ and $c_1=2$,
	the map is depicted in figure \ref{f11}.
	\begin{figure}[htbp]
		\centering
		\includegraphics[scale=0.7]{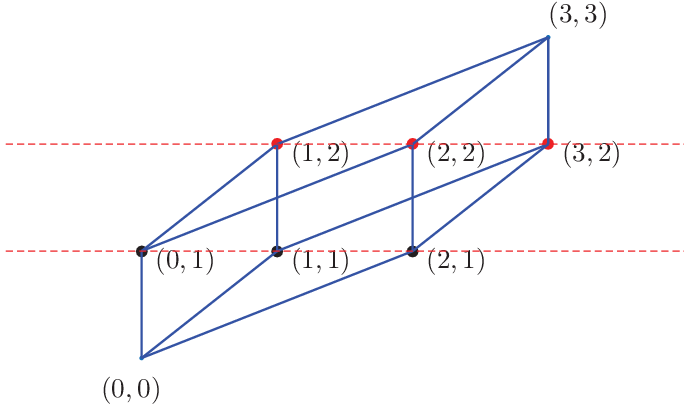}
		\caption{The $Q$-map with $Q=\{(0,1),(1,1),(2,1)\}$}
		\label{f11}
	\end{figure}
	Starting from the initial points of any two consecutive rows,
	we can obtain the whole grid of points by iterations of the maps.
	
\end{example}

\begin{example}\label{ex_1}
	Let $Q=\{(a_1,1),(b_1,2),(c_1,2)\}$ with $a_1,b_1,c_1\in\mathbb{Z}$, which is a $Q$-pin.
	First, we have $m=b_2+c_2=4$ and we can take $(V^{(1)},V^{(2)},V^{(3)},V^{(4)})\in\mathcal{U}_{3,4}$ as initial polygons.
	The condition $(P1)$ gives that
	$v_{i,1}$, $v_{i+a_1,2}$, $v_{i+b_1,3}$, $v_{i+a_1+b_1,4}$ are coplanar for all $i$.
	The condition $(P2)$ gives that
	$v_{i,1}$, $v_{i+a_1,2}$, $v_{i+c_1,3}$, $v_{i+a_1+c_1,4}$ are coplanar for all $i$.
	
	The map is defined as
	\begin{align*}
		v_{i+a_1+b_1+c_1,5}=
		\left<v_{i+a_1,1},v_{i+a_1+b_1,3},v_{i+a_1+c_1,3}\right>
		&\cap
		\left<v_{i+b_1,2},v_{i+a_1+b_1,3},v_{i+b_1+c_1,4}\right>\\
		&\cap
		\left<v_{i+c_1,2},v_{i+a_1+c_1,3},v_{i+b_1+c_1,4}\right>,
	\end{align*}	
	for all $i\in\mathbb{Z}$.
	In particular,
	taking $a_1=-1$, $a_2=1$, $a_3=3$,
	the map is shown in figure \ref{f1}.
	\begin{figure}[h]
		\centering
		\includegraphics[scale=0.7]{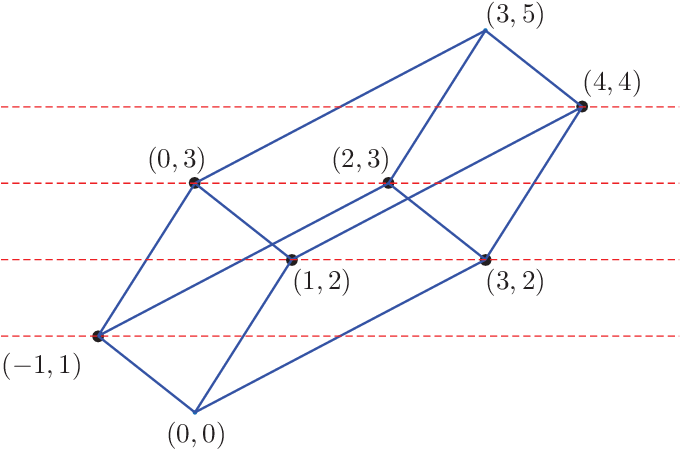}
		\caption{The $Q$-map with $Q=\{(-1,1),(1,2),(3,2)\}$}
		\label{f1}
	\end{figure}

\end{example}

\subsection{Higher-dimensional $Q$-maps}

Now we consider the case of dimension $N\geq 4$.
Suppose $V^j=(V^{(j+1)},\ldots,V^{(j+m)})=(v_{\cdot,j+1},\ldots,v_{\cdot,j+m})\in\mathcal{U}_{N,m}$ with $j\in\mathbb{Z}$ and $m=b_2+c_2$.
We need the following conditions in addition to $(P1)$ and $(P2)$ from before:

\vspace{10pt}
(C3) Points $v_{r+a}$, $v_{r+b}$, $v_{r+c}$, $v_{r+a+b}$, $v_{r+a+c}$ belong to a subspace $\mathbb{RP}^3$ for all $r$ with $j-a_2+1\leq r_2\leq j+b_2-a_2$;

(C4) Points $v_{r+b}$, $v_{r+c}$, $v_{r+a+b}$, $v_{r+a+c}$, $v_{r+b+c}$ belong to a subspace $\mathbb{RP}^3$ for all $r$ with $j-b_2+1\leq r_2\leq j$;

\vspace{10pt}
Define $X_{N,Q}\subset\mathcal{U}_{N,m}$ by $V^j\in X_{N,Q}$ iff $v_{i,j}=V_i^{(j)}$ satisfy $(P1)$, $(P2)$, $(C3)$, $(C4)$.
The maps $F: X_{N,Q}\rightarrow X_{N,Q}$ and $G: X_{N,Q}\rightarrow X_{N,Q}$ are defined using \eqref{def_1} and \eqref{def_2}.

\begin{remark}
According to \eqref{def_1},
$v_{r+a+b+c}$ is the intersection of three planes in $\mathbb{RP}^N$ $$\left<v_{r+a},v_{r+a+b},v_{r+a+c}\right>,\quad
\left<v_{r+b},v_{r+a+b},v_{r+b+c}\right>
,\quad
\left<v_{r+c},v_{r+a+c},v_{r+b+c}\right>.$$
When dimension $N\geq 4$,
these three planes intersect only when they are in one and the same $3$-dimensional subspace $\mathbb{RP}^3$.
Equivalently,
the six points 
$$v_{r+a},\quad
v_{r+b},\quad
v_{r+c},\quad
v_{r+a+b},\quad
v_{r+a+c},\quad
v_{r+b+c}
$$ appeared in \eqref{def_1} and \eqref{def_2} belong to a subspace $\mathbb{RP}^3$.
That is the reason why we need the conditions $(C3)$ and $(C4)$.

\end{remark}

\begin{proposition}
	The maps $F$ and $G$ map $X_{N,Q}$ to $X_{N,Q}$ and they are inverse to each other.
	
\end{proposition}
\begin{proof}
	Most of the proof is similar to that of Proposition \ref{pro_3}.
	What remains to verify is that the maps preserve $(C3)$ and $(C4)$.
	First,
	we prove $F(V^j)\in X_{N,Q}$.
	Let $V^j=(V^{(j+1)},\cdots,V^{(j+m)})\in X_{N,Q}$ and $F(V^j)=(V^{(j+2)},\cdots,V^{(j+m+1)})$.
	For $(C3)$,
	we should check it for $r_2+a_2+c_2=j+m+1$,
	that is $r_2=j+b_2-a_2+1$.
	According to the definition of $F$
	\begin{align*}
		v_{r+a+b+c}
		=\left<v_{r+a},v_{r+a+b},v_{r+a+c}\right>
		\cap
		\left<v_{r+b},v_{r+a+b},v_{r+b+c}\right>
		\cap
		\left<v_{r+c},v_{r+a+c},v_{r+b+c}\right>,
	\end{align*}
	the points $v_{r+a}$, $v_{r+a+b}$, $v_{r+a+c}$, $v_{r+a+b+c}$ are coplanar;
	$v_{r+b}$, $v_{r+a+b}$, $v_{r+b+c}$, $v_{r+a+b+c}$ are coplanar;
	$v_{r+c}$, $v_{r+a+c}$, $v_{r+b+c}$, $v_{r+a+b+c}$ are coplanar.
	The three planes intersect (at $v_{r+a+b+c}$),
	so the points $v_{r+a}$, $v_{r+b}$, $v_{r+c}$, $v_{r+a+b}$, $v_{r+a+c}$ belong to a subspace $\mathbb{RP}^3$,
	where $r_2=j+b_2-a_2+1$.
	A similar argument shows $(C4)$ with $r_2=j+1$.
	
	The conclusion $G(V^j)\in X_{N,Q}$ can be proved using the method above.
	Finally, we can prove that $F$ and $G$ are inverse to each other similar to proposition \ref{pro_3}.
\end{proof}
\begin{remark}
According to the proposition above,
$(P1)$, $(P2)$, $(C3)$, $(C4)$ are also sufficient for the existence of reduced $Q$-nets,
since we could apply $F$ and $G$ to an initial element $V^j$ in $X_{N,Q}$ iteratively to build the full $v_{i,j}$ array.
\end{remark}

\begin{example}
	Let's consider the same $Q$-pin $Q=\{(a_1,1),(b_1,2),(c_1,2)\}$ as in Example \ref{ex_1},
	but here we take dimension $N=4$.
	First, we have $m=b_2+c_2=4$ and we can take $(V^{(1)},V^{(2)},V^{(3)},V^{(4)})\in\mathcal{U}_{4,4}$ as initial polygons.
	The conditions $(P1)$ and $(P2)$ are the same as Example \ref{ex_1}.
	
	In addition,
	the condition $(C_3)$ gives that
	$$v_{i+a_1,1},\quad v_{i+b_1,2},\quad v_{i+c_1,2},\quad v_{i+a_1+b_1,3},\quad v_{i+a_1+c_1,3}$$
	belong to a $3$-dimensional subspace for all $i\in\mathbb{Z}$.
	The condition $(C_4)$ gives that
	$$v_{i+b_1,2},\quad v_{i+c_1,2},\quad v_{i+a_1+b_1,3},\quad v_{i+a_1+c_1,3},\quad 
	v_{i+b_1+c_1,4}$$ 
	belong to a $3$-dimensional subspace for all $i\in\mathbb{Z}$.

\end{example}

\section{Explicit formulas for one specific example}\label{sec_expl}

In this section,
we focus on one specific map \eqref{map1}.
Let $$Q=\{(-1,1),(0,1),(1,1)\}$$ be a $Q$-pin.
Suppose that $(v_i^{1},v_i^{2})$, $i\in\mathbb{Z}$ are two twisted $N$-gons in $\mathbb{RP}^3$ with the same monodromy $M$.
Let $\mathscr{P}_{3,N}$ be the space of all pairs of twisted $N$-gons $(v_i^1,v_i^2)$ in $\mathbb{RP}^3$ with the same monodromy,
where $(v_i^1,v_i^2)$ are in general position, that is,  the four points
$$v_{i-1}^1,\quad v^2_{i-1},\quad v_i^1,\quad v_i^2$$ 
don't belong to one and the same two-dimensional subspace in $\mathbb{RP}^3$.
The corresponding $Q$-map $T_Q:\mathscr{P}_{3,N}\rightarrow\mathscr{P}_{3,N}$ is  defined as (see Figure \ref{f0})
\begin{align*}
T_Q\left(v_i^{1},v_i^{2}\right)=\left(v_i^{2},v_i^{3}\right),
\end{align*}
where
\begin{align}\label{def_example}
\begin{split}
	v_i^{3}:=\left(v_{i-1}^{1},\,v_{i-1}^{2},\,v_{i}^{2}\right)
	\cap
	\left(v_{i}^{1},\,v_{i-1}^{2},\,v_{i+1}^{2}\right)
	\cap
	\left(v_{i+1}^{1},\,v_{i}^{2},\,v_{i+1}^{2}\right).
\end{split}
\end{align}

\subsection{Polygons and difference equations}

Consider the following difference equations
\begin{align}\label{recur_eq0}
	\begin{split}
		V_{i+1}^{1}&=a_iV_{i-1}^{1}+b_iV_{i-1}^{2}+V_{i}^{1}+c_iV_i^{2},\\
		V_{i+1}^{2}&=f_iV_{i-1}^{1}+g_iV_{i-1}^{2}+V_{i}^{1}+h_iV_i^{2},
	\end{split}
\end{align}
where $a_i$, $b_i$, $c_i$, $f_i$, $g_i$, $h_i$ are arbitrary $N$-periodic sequences on the index $i$, such that $a_ig_i-b_if_i\neq 0$, $\forall i\in\mathbb{Z}$.

The space of solutions of \eqref{recur_eq0} is $4$-dimensional,
since any solution is determined by the initial conditions $(V_{-1}^1 ,V_{-1}^2,V_0^1,V_0^2)$.
Thus,
we often understand $V_i^1$, $V_i^2$ as vectors in $\mathbb{R}^4$.
The $N$-periodicity then implies that there exists a matrix $M\in\mathrm{SL}(4,\mathbb{R})$ called the monodromy matrix, such that $V_{i+N}^j=MV_i^j$, $j=1,2$. 

Let $(V_{-1}^1 ,V_{-1}^2,V_0^1,V_0^2)$ be  the initial value,
such that $\det(V_{-1}^1 ,V_{-1}^2,V_0^1,V_0^2)\neq 0$,
that is the four vectors are linearly independent.
According to \eqref{recur_eq0},
we have 
$$\det(V_i^1,V_i^2,V_{i+1}^1,V_{i+1}^2)=
\left(a_ig_i-b_if_i\right)
\det(V_{i-1}^1 ,V_{i-1}^2,V_i^1,V_i^2).$$
Since $a_ig_i-b_if_i\neq 0$ for all $i\in\mathbb{Z}$,
it is obvious that $(V_{i-1}^1 ,V_{i-1}^2,V_i^1,V_i^2)$  
are also linearly independent for all $i\in\mathbb{Z}$.

\begin{proposition}
	The space $\mathscr{P}_{3,N}$ is isomorphic to the space of the equation \eqref{recur_eq0}.
	
\end{proposition}
\begin{proof}
	
	Let $(v_i^{1},\,v_i^{2})$ be a pair of twisted $N$-gons in $\mathscr{P}_{3,N}$ with the monodromy $M$.
	Let $\tilde{V}_i^{1}$ and $\tilde{V}_i^{2}\in\mathbb{R}^4$ be two arbitrary lifts of $v_i^{1},\,v_i^{2}$ such that $\tilde{V}_{i+N}^{1}=M\tilde{V}_i^{1}$ and $\tilde{V}_{i+N}^{2}=M\tilde{V}_i^{2}$.
	The general position property implies  $\det(\tilde{V}_{i-1}^{1},\tilde{V}_{i-1}^{2},\tilde{V}_{i}^{1},\tilde{V}_{i}^{2})\neq 0$ for all $i\in\mathbb{Z}$.
	Taking the four vectors $(\tilde{V}_{i-1}^{1},\tilde{V}_{i-1}^{2},\tilde{V}_{i}^{1},\tilde{V}_{i}^{2})$ as a basis of $\mathbb{R}^4$,
	the vectors
	$\tilde{V}_{i+1}^{1}$, $\tilde{V}_{i+1}^{2}$ have the following expressions
	\begin{align*}
		\begin{split}
			\tilde{V}_{i+1}^{1}&=\tilde{a}_i\tilde{V}_{i-1}^{1}+\tilde{b}_i\tilde{V}_{i-1}^{2}+\tilde{d}_i\tilde{V}_{i}^{1}+\tilde{c}_i\tilde{V}_i^{2},\\
			\tilde{V}_{i+1}^{2}&=\tilde{f}_i\tilde{V}_{i-1}^{1}+\tilde{g}_i\tilde{V}_{i-1}^{2}+\tilde{w}_i\tilde{V}_{i}^{1}+\tilde{h}_i\tilde{V}_i^{2},
		\end{split}
	\end{align*}
	for some sequences $\tilde{a}_i$, $\tilde{b}_i$, $\tilde{c}_i$, $\tilde{d}_i$, $\tilde{f}_i$, $\tilde{g}_i$, $\tilde{h}_i$, $\tilde{w}_i$.
	Here it is not difficult to prove that these sequences are $N$-periodic according to the conditions $\tilde{V}_{i+N}^{1}=M\tilde{V}_i^{1}$ and $\tilde{V}_{i+N}^{2}=M\tilde{V}_i^{2}$, for all $i\in\mathbb{Z}$.
	Now we rescale: $V_i^j=\frac{1}{t_i^j}\tilde{V}_i^j$, $j=1,2$, by using two sequences $\{t_i^1\}_{i\in\mathbb{Z}}$, $\{t_i^2\}_{i\in\mathbb{Z}}$, which are determined by
	\begin{align}\label{eq_t}
		t_{i+1}^{1}=\tilde{d}_it_i^1,\quad
		t_{i+1}^2=\tilde{w}_it_i^1,\quad
		\forall i\in\mathbb{Z}.
	\end{align}
	It is not hard to prove that the equation for $V$ is just
	\eqref{recur_eq0},
	where
	\begin{align}\label{eq_ab}
		\begin{split}
			a_i=\frac{t_{i-1}^1}{t_{i+1}^1}\tilde{a}_i,\quad
			b_i=\frac{t_{i-1}^2}{t_{i+1}^1}\tilde{b}_i,\quad
			c_i=\frac{t_{i}^2}{t_{i+1}^1}\tilde{c}_i,\\
			f_i=\frac{t_{i-1}^1}{t_{i+1}^2}\tilde{f}_i,\quad
			g_i=\frac{t_{i-1}^2}{t_{i+1}^2}\tilde{g}_i,\quad
			h_i=\frac{t_{i}^2}{t_{i+1}^2}\tilde{h}_i.
		\end{split}
	\end{align}
	Suppose that $t_{0}^1=\alpha\neq 0$ is fixed and $\tilde{d}_i\tilde{w}_i\neq 0$, $i\in\mathbb{Z}$,
	then all $t_i^j$, $i\in\mathbb{Z}$, $j=1,2$ can be determined uniquely by \eqref{eq_t}.
	Though the initial value $t_0^1=\alpha$ is arbitrary,
	it doesn't affect the values of
	$a_i$, $b_i$, $c_i$, $f_i$, $g_i$, $h_i$ according to \eqref{eq_ab},
	so the values of
	$a_i$, $b_i$, $c_i$, $f_i$, $g_i$, $h_i$ are determined uniquely.
	
	Next we will prove the sequences $a_i$, $b_i$, $c_i$, $f_i$, $g_i$, $h_i$ are $N$-periodic.
	Before doing so,
	we need to prove $t_i^1$, $t_i^2$ are quasi-periodic sequences with the same monodromy.
	According to the first equation of \eqref{eq_t},
	we get
	\begin{align*}
		t_{i+N}^1=\tilde{d}_{i+N-1}\cdots\tilde{d}_{i}t_i^1.
	\end{align*}
	Since $\tilde{d}_i$ is $N$-periodic,
	so $t_i^1$ is quasi-periodic, that is $t_{i+N}^1=Dt_i^1$, for all $i\in\mathbb{Z}$, where $D_0=\tilde{d}_1\cdots \tilde{d}_N$ is independent of $i$.
	Besides, according to the second equation of \eqref{eq_t}, we have
	\begin{align*}
		t_{i+N}^2=\tilde{w}_{i+N-1}t_{i+N-1}^1=\tilde{w}_{i+N-1}\tilde{d}_{i+N-2}\cdots\tilde{d}_{i-1}t_{i-1}^1
		=\tilde{w}_{i+N-1}\tilde{d}_{i+N-2}\cdots\tilde{d}_{i-1}\frac{1}{\tilde{w}_{i-1}}t_i^2.
	\end{align*}
	Since $\tilde{d}_i$, $\tilde{w}_i$ are $N$-periodic,
	we prove $t_i^2$ is a quasi-periodic sequence possessing the same monodromy $D_0$ as $t_i^1$.
	
	Using the quasi-periodicity of $t_i^1$, $t_i^2$ and the expressions \eqref{eq_ab},
	we prove that
	$a_i$, $b_i$, $c_i$, $f_i$, $g_i$, $h_i$ are $N$-periodic.
	This means that any pairs of twisted $N$-gons in $\mathscr{P}_{3,N}$ have a lift to $\mathbb{R}^4$ satisfying \eqref{recur_eq0}.

	Furthermore,
	if $(v_i^1,v_i^2)$ and  $(u_i^1,u_i^2)$, $i\in\mathbb{Z}$ are two projectively equivalent twisted $N$-gons,
	then
	they correspond to the same equation \eqref{recur_eq0}.
	Indeed, there exists $A\in\mathrm{SL}(4,\mathbb{R})$ such that $A(v_i^j)=u_i^j$, $j=1,2$.
	Equivalently,
	there exist lifts $(U^1,U^2)$, s.t. $A(V_i^j)=U_i^j$, $j=1,2$, $i\in\mathbb{Z}$. 
	The sequences $\{U_i^j\}_{i\in\mathbb{Z}}$, $j=1,2$ then 
	satisfy the same equation \eqref{recur_eq0}.
	
	Conversely,
	let $(V_i^1,V_i^2)$ be a solution of \eqref{recur_eq0}.
	Since $a_i$, $b_i$, $c_i$, $f_i$, $g_i$, $h_i$ are periodic,
	it follows that \eqref{recur_eq0} defines a pair of twisted $N$-gons in $\mathscr{P}_{3,N}$.
	A choice of initial conditions $(V_{-1}^1,V_{-1}^2,V_0^1,V_0^2)$ fixes a pair of polygons and a different choice yields a projectively equivalent one.
\end{proof}
\begin{remark}
In general,
the lift of $(v^1,v^2)$ to $(V^1,V^2)\subset\mathbb{R}^4$ satisfying \eqref{recur_eq0} is not unique,
because the sequences $t_i^1$, $t_i^2$ defined in \eqref{eq_t} are not unique.
However,
the $N$-periodic coordinates $a_i$, $b_i$, $c_i$, $f_i$, $g_i$, $h_i$ 
do not depend on the choice of the lift coefficients $t_i^1$, $t_i^2$.
These coordinates are analogs of the cross-ratio coordinates $x_i$, $y_i$ in \cite{2010The}.
\end{remark}

\begin{proposition}
	The coordinates $a_i$, $b_i$, $c_i$, $f_i$, $g_i$, $h_i$ can be expressed using cross-ratios: 
	\begin{align}\label{eq_cr}
		\begin{split}
			a_i=&-\left[V_{i-1}^1,
			L_{-1,0}^{11}\cap P_{-2,-2,-1}^{122},
			V_i^1,
			L_{-1,0}^{11}\cap P_{-1,0,1}^{221}\right],\\
			f_i=&-\left[V_{i-1}^1,
			L_{-1,0}^{11}\cap P_{-2,-2,-1}^{122},
			V_i^1,
			L_{-1,0}^{11}\cap P_{-1,0,1}^{222}\right],\\
			c_i=&
			-
			\left[V_{i-1}^1,
			L_{-1,0}^{12}\cap P_{-1,0,1}^{211},
			V_i^2,
			L_{-1,0}^{12}\cap P_{-2,-2,-1}^{122}\right]\cdot a_i,\\
			b_i=&
			-\left[V_{i-1}^2,
			L_{-1,0}^{21}\cap P_{-2,-2,-1}^{121},
			V_i^1,
			L_{-1,0}^{21}\cap P_{-1,0,1}^{121}\right]
			\cdot c_{i-1},\\
			g_i=&
			-\left[V_{i-1}^2,
			L_{-1,0}^{21}\cap P_{-2,-2,-1}^{121},
			V_i^1,
			L_{-1,0}^{21}\cap P_{-1,0,1}^{122}\right]
			\cdot c_{i-1},\\
			h_i=&
			\left[V_{i}^1,
			L_{0,0}^{12}\cap P_{-1,-1,1}^{122},
			V_i^2,
			L_{0,0}^{12}\cap P_{-1,-1,1}^{121}\right]
			\cdot c_{i},
		\end{split}
	\end{align}
	where $L_{i_1,i_2}^{j_1,j_2}$ represents the line $\left(V_{i+i_1}^{j_1},V_{i+i_2}^{j_2}\right)$ and $P_{i_1,i_2,i_3}^{j_1,j_2,j_3}$ represents the plane passing through three points $\left(V_{i+i_1}^{j_1},V_{i+i_2}^{j_2},V_{i+i_3}^{j_3}\right)$ for all $j_1$, $j_2$, $j_3\in\{1,2\}$.
\end{proposition}
By this definition these coordinate are projectively invariant.
\begin{proof}
	Here we just consider the case of $a_i$ and the others can be proved similarly.
	A simple computation shows that
	\begin{align*}
		L_{-1,0}^{11}\cap P_{-2,-2,-1}^{122}=V_{i-1}^1-V_i^1,\quad
		L_{-1,0}^{11}\cap P_{-1,0,1}^{221}=a_iV_{i-1}^1+V_i^1.
	\end{align*}
	Recall that given four coplanar vectors $X,Y,Z,W$ such that
	\begin{align*}
		Y=\lambda_1X+\lambda_2Z,\quad
		W=\mu_1 X+\mu_2 Z,
	\end{align*}
	where $\lambda_i$, $\mu_i$, $i=1,2$ are constants,
	the cross-ratio of the lines spanned by these vectors is given by
	\begin{align*}
		[X,Y,Z,W]=\frac{\lambda_2\mu_1}{\lambda_1\mu_2}.
	\end{align*}
	Using this definition, 
	we get the cross-ratio expression of $a_i$.
\end{proof}

\subsection{Explicit formulas for the map}
Now we consider the discrete time evolution of the pair of twisted $N$-gons
\begin{align*}
	T_Q(v_i^1,v_i^2)=(v_i^2,v_i^3),\quad
	\forall i\in\mathbb{Z},
\end{align*}
where $v_i^3$ is defined in \eqref{def_example}.
\begin{proposition}
	The vectors $(V_{i-1}^3,V_i^3)$ have the following expressions
	\begin{align}\label{eq_33}
		\begin{split}
			V_{i-1}^3&=\frac{1}{\lambda_{i-1}}\left(({a_{i-1}h_{i-1}}-c_{i-1}f_{i-1})V_{i-1}^2+f_{i-1}V_{i}^1-f_{i-1}V_i^2\right),\\
			V_i^3&=\frac{1}{\lambda_i}\left(f_i(a_i-f_i)V_{i-1}^1+f_i(b_i-g_i)V_{i-1}^2+h_i(a_i-f_i)V_i^2\right).
		\end{split}
	\end{align}
	where $\lambda_i$ is a scalar function
	\begin{align}\label{eq_lambda}
		\lambda_i={{f_i}\left(a_i+b_i-g_i-f_i\right)+{h_i}(a_i-f_i)}.
	\end{align}
	
\end{proposition}
\begin{proof}
	According to \eqref{def_example},
	we have
	\begin{align}\label{eq_mi}
		\begin{split}
			V_i^3=&x_i^{11}V_{i-1}^{1}+x_i^{12}V_{i-1}^{2}+x_i^{13}V_{i}^{2}\\
			=&
			x_i^{21}V_{i}^{1}+x_i^{22}V_{i-1}^{2}+x_i^{23}V_{i+1}^{2}\\
			=&
			x_i^{31}V_{i+1}^{1}+
			x_i^{32}V_{i}^{2}+
			x_i^{33}V_{i+1}^{2},
		\end{split}
	\end{align}
	where $x_i^{kl}$, $k,l=1,2,3$, are undetermined scalar functions.
	Next we use \eqref{recur_eq0} to omit $V_{i+1}^1$, $V_{i+1}^2$ in \eqref{eq_mi}.
	Then, comparing the coefficients of $V_{i-1}^1$, $V_{i-1}^2$, $V_i^1$, $V_i^2$,
	we get $8$ equations
	\begin{align*}
		x_i^{22}+x_i^{23}=0,\quad
		h_ix_i^{23}-x_i^{13}=0,\quad
		x_i^{32}+x_i^{33}=0,\quad
		f_ix_i^{23}-x_i^{11}=0,\\
		g_ix_i^{23}-x_i^{12}+x_i^{21}=0,\quad
		c_ix_i^{32}+h_ix_i^{33}-x_i^{13}+x_i^{31}=0,\\
		a_ix_i^{32}+f_ix_i^{33}-x_i^{11}=0,\quad
		b_ix_i^{32}+g_ix_i^{33}-x_i^{12}=0.
	\end{align*}
	After solving them,
	we get a solution of $x_i^{kl}$ up to a nonzero multiplier.
	Substituting the solution into the first row of \eqref{eq_mi},
	we get the second identity of \eqref{eq_33}.
	Besides,
	according to the third row of \eqref{eq_mi},
	we get 
	\begin{align*}
		V_{i-1}^3=x_{i-1}^{31}V_{i}^{1}+
		x_{i-1}^{32}V_{i-1}^{2}+
		x_{i-1}^{33}V_{i}^{2},
	\end{align*}
	after taking a shift $i\rightarrow i-1$.
	Substituting the solutions of $x_i^{kl}$ into the equation above,
	we get the first identity of \eqref{eq_33}.
	Finally, $\lambda_i$ can be determined by the condition $V_i^2$, $V_i^3$, $i\in\mathbb{Z}$ satisfy a relation similar to \eqref{recur_eq0}
	\begin{align*}
		\begin{split}
			V_{i+1}^{2}&=\hat{a}_iV_{i-1}^{2}+\hat{b}_iV_{i-1}^{3}+V_{i}^{2}+\hat{c}_iV_i^{3},\\
			V_{i+1}^{3}&=\hat{f}_iV_{i-1}^{2}+\hat{g}_iV_{i-1}^{3}+V_{i}^{2}+\hat{h}_iV_i^{3},
		\end{split}
	\end{align*}
	where $\hat{a}_i$, $\ldots$, $\hat{h}_i$ are undetermined functions.
\end{proof}

Let's introduce the index $j$ to represent the discrete time.
Then the relation \eqref{recur_eq0} can be written as 
\begin{align}\label{recur_eq}
	\begin{split}
		V_{i+1}^{j+1}&=a_i^jV_{i-1}^{j+1}+b_i^jV_{i-1}^{j+2}+V_{i}^{j+1}+c_i^jV_i^{j+2},\\
		V_{i+1}^{j+2}&=f_i^jV_{i-1}^{j+1}+g_i^jV_{i-1}^{j+2}+V_{i}^{j+1}+h_i^jV_i^{j+2},
	\end{split}
\end{align}
where $a_i^j$, $b_i^j$, $c_i^j$, $f_i^j$, $g_i^j$, $h_i^j$ are arbitrary $N$-periodic sequences on the index $i$.

\begin{proposition}
	Under the coordinates $a_i^j$, $b_i^j$, $c_i^j$, $f_i^j$, $g_i^j$, $h_i^j$,
	the map is given by the formulas
	\begin{align}\label{eq_sys}
		\begin{split}
			a_i^{j+1}=c_{i-1}^j-
			\frac{h_{i-1}^j}{f_{i-1}^j}a_{i-1}^j+
			\frac{a_i^jg_i^j-b_i^jf_i^j}{a_i^j-f_i^j}{a_i^j}
			,
			\quad
			b_i^{j+1}=\frac{\lambda_{i-1}^j}{f_{i-1}^j},\quad
			c_i^{j+1}=\frac{\lambda_i^j}{a_i^j-f_i^j},\\
			f_i^{j+1}=
			\frac{a_{i+1}^j-f_{i+1}^j}{\lambda_{i+1}^j}\left(
			\left(c_{i-1}^jf_{i-1}^j-{a_{i-1}^jh_{i-1}^j}\right)\frac{f_{i+1}^j+h_{i+1}^j}{f_{i-1}^j}+h_{i+1}^j\frac{a_i^jg_i^j-b_i^jf_i^j}{a_i^j-f_i^j}
			\right),\\
			g_i^{j+1}=\frac{\lambda_{i-1}^j}{\lambda_{i+1}^jf_{i-1}^j}(a_{i+1}^j-f_{i+1}^j)(f_{i+1}^j+h_{i+1}^j),\quad
			h_i^{j+1}=\frac{\lambda_i^j}{\lambda_{i+1}^j}\frac{a_{i+1}^j-f_{i+1}^j}{a_{i}^j-f_{i}^j}h_{i+1}^j,
		\end{split}
	\end{align}
	where $\lambda_i^j$ is defined in \eqref{eq_lambda}.
	
\end{proposition}

\begin{proof}
Introduce the function
\begin{align*}
	\Phi_i^j:=\left(V_{i-1}^{j+1},\,V_{i-1}^{j+2},\,V_i^{j+1},\,V_i^{j+2}\right)^T.
\end{align*}
According to \eqref{recur_eq} and \eqref{eq_33},
we get the following linear problem
\begin{align}\label{eq_lax}
	\Phi_{i+1}^j=L_i^j\Phi_i^j,\quad
	\Phi_{i}^{j+1}=M_i^j\Phi_i^j,
\end{align}
where 
\begin{align*}
	L_i^j:=\begin{pmatrix}
		0&0&1&0\\
		0&0&0&1\\
		a_i^j&b_i^j&1&c_i^j\\
		f_i^j&g_i^j&1&h_i^j
	\end{pmatrix},\quad
	M_i^j:=
	\begin{pmatrix}
		0&1&0&0\\
		0&\frac{\alpha_{i-1}^j}{\lambda_{i-1}^j}&\frac{f_{i-1}^j}{\lambda_{i-1}^j}&-\frac{f_{i-1}^j}{\lambda_{i-1}^j}\\
		0&0&0&1\\
		\frac{\beta_i^j}{\lambda_i^j}&\frac{\gamma_i^j}{\lambda_i^j}&0&\frac{\eta_i^j}{\lambda_i^j}
	\end{pmatrix}.
\end{align*}
Here
\begin{align*}
	\alpha_n^j=a_{n}^jh_{n}^j-c_{n}^jf_{n}^j,\quad
	\beta_i^j=f_i^j\left(a_i^j-f_i^j\right),\quad
	\gamma_i^j=f_i^j\left(b_i^j-g_i^j\right),\\
	\eta_i^j=h_i^j\left(a_i^j-f_i^j\right),\quad
	\lambda_i^j=f_i^j\left(a_i^j+b_i^j-g_i^j-f_i^j\right)+h_i^j(a_i^j-f_i^j).
\end{align*}
From the compatibility condition
\begin{align*}
	L_i^{j+1}M_i^j=M_{i+1}^jL_i^j,
\end{align*}
we could get \eqref{eq_sys}.
\end{proof}

\section{Poisson structure}\label{sec_refa}
In this section,
the detailed proofs of conclusions in Section \ref{subsec_re} are demonstrated,
and we obtain the Poisson structure of the system \eqref{eq_sys}.

\subsection{Difference operators}
Given a pair of twisted $N$-gons $(v^1,v^2)\in\mathscr{P}_{3,N}$,
let $(V^1,V^2)\subset\mathbb{R}^4$ be an arbitrary lift such that
\begin{align*}
V_{i+N}^1=MV_i^1,\quad
V_{i+N}^2=MV_i^2,
\end{align*}
where $M\in\mathrm{SL}(\mathbb{R},4)$ is the monodromy.
We write $(V^1,V^2)$ as an infinite-dimensional sequence $\Phi$
\begin{align*}
\Phi_k=\begin{cases}
V_i^1,\quad k=2i,\\
V_{i}
^2,\quad
k=2i+1,
\end{cases}
\end{align*}
that is
\begin{align}\label{def_phi}
\Phi=\begin{pmatrix}
\cdots,&
V_{i-1}^1,&
V_{i-1}^2,&
V_i^1,&
V_i^2,&
V_{i+1}^1,&
V_{i+1}^2,&
\cdots
\end{pmatrix}^T.
\end{align}
Note that $T_Q(v^1,v^2)=(v^2,v^3)\in\mathscr{P}_{3,Q}$.
Let $\tilde{\Phi}$ denote the corresponding  sequence 
\begin{align}\label{def_phi2}
	\tilde{\Phi}=\begin{pmatrix}
		\cdots,&
		V_{i-1}^2,&
		V_{i-1}^3,&
		V_i^2,&
		V_i^3,&
		V_{i+1}^2,&
		V_{i+1}^3,&
		\cdots
	\end{pmatrix}^T,
\end{align}
where $V_i^3$ is defined in \eqref{def_example}.

\begin{proposition}\label{pro_equi}
	The two twisted $N$-gons $(v^1,v^2)$ in $\mathbb{RP}^3$ modulo the projective transformations can be encoded by means of difference operators in $\PBDO_{2N}(\J)/H\tilde{\times}H$ with $\J=\{0,1,2,3,4\}$, that is
	\begin{align}\label{op_D}
		\D=a^0+a^1T+a^2T^2+a^3T^{3}+a^4T^4,
	\end{align}
	where $a^j$, $j=0,1,2,3,4$ are sequences such that $a_k^0\neq 0$, $a^4_k\neq 0$ for any $k\in\mathbb{Z}$. 
\end{proposition}

\begin{proof}
We can treat $\Phi$ in \eqref{def_phi} as a new twisted $2N$-gon in $\mathbb{RP}^3$.
Then the conclusion follows from proposition \ref{pro_equi0}.
\end{proof}

Now we give a Lemma which connects the geometric characteristic of the $Q$-map and difference operators. 
\begin{lemma}\label{lemma_d123}

The vectors $\Phi$, $\tilde{\Phi}$ defined in \eqref{def_phi} and \eqref{def_phi2} satisfy the following difference equations
\begin{align}\label{eq_dd}
	\begin{split}
	\D_1\tilde{\Phi}&=\D_2\Phi,\\
	\tilde{\Phi}&=\D_3\Phi,
	\end{split}
\end{align}
where $\D_i$, $i=1,2,3$ are $2N$-periodic difference operators such that
$\D_1$ is supported in $\{2,3\}$,
$\D_2$ is supported in $\{0,1\}$,
and $\D_3=f+gT+hT^2$ is supported in $\{0,1,2\}$, such that
\begin{align*}
	f_{2i}=h_{2i}=0,\quad \forall i\in\mathbb{Z}.
\end{align*}
\end{lemma}

\begin{proof}
According to the definition of the $Q$-map,
we know that the four points $V_{i-1}^1$, $V_{i-1}^2$, $V_{i}^2$, $V_i^3$ are coplanar,
while $V_{i-1}^2$, $V_i^1$, $V_{i}^3$, $V_{i+1}^2$ are also coplanar.
Therefore,
we have
\begin{align*}
\alpha_iV_i^2+\beta_iV_i^3&=\gamma_i V_{i-1}^1+\delta_i V_{i-1}^2,\\
\epsilon_i V_i^3+\zeta_i V_{i+1}^2&=\eta_i V_{i-1}^2+\theta_iV_i^1,
\end{align*}
which can be written using difference operators as
\begin{align*}
\D_1\tilde{\Phi}=\D_2\Phi,
\end{align*}
where
\begin{align*}
\D_1=aT^2+bT^3,\quad
\D_2=c+dT,
\end{align*}
with
\begin{align*}
(a_k,b_k,c_k,d_k)=\begin{cases}
(\alpha_i,\beta_i,\gamma_i,\delta_i),\quad k=2i,\\
(\epsilon_i,\zeta_i,\eta_i,\theta_i),\quad k=2i+1.
\end{cases}
\end{align*}
Furthermore,
note that $V_i^2$, $V_{i+1}^1$, $V_{i+1}^2$, $V_i^3$ are coplanar,
and $T_QV_i^1=\rho_i V_i^2$,
which yield
\begin{align*}
T_QV_i^1&=\rho_iV_i^2,\\
V_i^3&=\sigma_iV_i^2+\lambda_iV_{i+1}^1+\mu_iV_{i+1}^2.
\end{align*}
Using difference operators,
these two equations above can be rewritten as
\begin{align*}
	\tilde{\Phi}=\D_3\Phi,\quad
	\D_3=f+gT+hT^3,
\end{align*}
where
\begin{align*}
(f_k,g_k,h_k)=\begin{cases}
	(0,\rho_i,0),\quad k=2i,\\
	(\sigma_i,\lambda_i,\mu_i),\quad k=2i+1.
\end{cases}
\end{align*}
\end{proof}
Eliminating $\tilde{\Phi}$ in \eqref{eq_dd},
we get
\begin{align}\label{eq_ph}
\D_1\D_3\Phi=\D_2\Phi.
\end{align}

\begin{lemma}\label{lemma_fac}
For a periodic difference operator
$\D$ with the form $xT^2+yT^3+zT^4+wT^5$
such that
$x_{2i}=w_{2i+1}=0$, $\forall i\in\mathbb{Z}$,
there is a unique factorization 
\begin{align*}
\D=\D_1\D_3,
\end{align*}
up to the action
\begin{align}\label{eq_13}
	\D_1\mapsto \D_1\beta^{-1},\quad
	\D_3\mapsto\beta\D_3,
\end{align}
where
$\D_1$, $\D_3$ satisfy the conditions in lemma \ref{lemma_d123},
and  $\beta$ is a quasi-periodic sequence.

\end{lemma}
\begin{proof}

The matrix form of $\D$ is
\begin{align*}
	\begin{pmatrix}
		\ddots&\ddots&\ddots&\ddots&\\
		&0&y_{2i}&z_{2i}&w_{2i}&\\
		&&x_{2i+1}&y_{2i+1}&z_{2i+1}&0&\\
		&&&\ddots&\ddots&\ddots&\ddots&
	\end{pmatrix}.
\end{align*}
By direct computation,
it admits the following factorization $\D=\D_1\D_3,$ where
\begin{align*}
\D_1=
\begin{pmatrix}
		\ddots&\ddots&\\
&y_{2i}-\frac{x_{2i+1}}{y_{2i+1}}z_{2i}&z_{2i}&\\
&&y_{2i+1}&z_{2i+1}-\frac{w_{2i}}{z_{2i}}y_{2i+1}&\\
&&&\ddots&\ddots&
\end{pmatrix},
\end{align*}
and
\begin{align*}
	\D_3=
	\begin{pmatrix}
		\ddots&\ddots&\ddots&\\
		&0&1&0&\\
		&&\frac{x_{2i+1}}{y_{2i+1}}&1&\frac{w_{2i}}{z_{2i}}&\\
		&&&0&1&0&\\
		&&&&\ddots&\ddots&\ddots&
	\end{pmatrix}.
\end{align*}
In addition,
the new sequences in \eqref{eq_13} also satisfy the factorization equation.

Now we prove the uniqueness.
Suppose that $\D=\hat{\D}_1\hat{\D}_3$ is another factorization.
Let $\ker^L(\D)$ denotes the left kernel of $\D$,
which is the set of  $x\in\mathbb{R}^\infty$ such that
\begin{align*}
x^T\D=0.
\end{align*}
For the difference operator $\hat{\D}_3$, satisfying the conditions in lemma \ref{lemma_d123},
its left kernel is empty.
This can be proved by solving the equation $x^T\hat{\D}_3=0$.
Therefore,
we have $$\ker^L(\D)=\ker^L(\hat{\D}_1)=\ker^L(\D_1).$$
Since $\D_1$ and $\hat{\D}_1$ have the same support $\{2,3\}$,
we have $\D_1=\hat{\D}_1\alpha$ for a quasi-periodic sequence $\alpha$ (a similar conclusion for right kernel appeared in \cite[Lemma 4.8]{refactorization2020}).
Finally,
the uniqueness of $\D_3$ up to the action \eqref{eq_13} can be obtained from the identity $\D_3=\D_1^{-1}\D$.
\end{proof}

\begin{theorem}\label{pro_gon}
	There is a one-to-one correspondence between the pair of twisted $N$-gons $(v^1,v^2)\in \mathscr{P}_{3,N}$ (modulo projective transformations),
	and three $2N$-periodic difference operators: $\D_1$, $\D_2$, $\D_3$ as in Lemma \ref{lemma_d123}, considered up to the action 
	\begin{align}\label{eq_map}
		\D_1\mapsto \alpha\D_1\beta^{-1},\quad
		\D_2\mapsto\alpha\D_2\gamma^{-1},\quad
		\D_3\mapsto\beta\D_3\gamma^{-1},
	\end{align}
	where $\alpha$, $\beta$, $\gamma$ are quasi-periodic sequences with the same monodromy.
\end{theorem}

\begin{proof}

Firstly,
given three $2N$-periodic difference operators $\D_j$, $j=1,2,3$, as in Lemma \ref{lemma_d123},
we can construct a vector $\Phi$ satisfying the equation \eqref{eq_ph}
\begin{align}\label{eq_linear}
\left(\D_1\D_3-\D_2\right)\Phi=0,
\end{align}
where $\D_1\D_3-\D_2$ is a difference operator supported in $\{0,1,2,3,4,5\}$.
Suppose that $\D_1=aT^2+bT^3$,
$\D_2=c+dT$ and $\D_3=f+gT+hT^2$ with
 $f_{2i}=h_{2i}=0$ for all $i\in\mathbb{Z}$.
Then the matrix form of the operator $$\D_1\D_3-\D_2=bh_3T^5+(bg_3+ah_2)T^4+(bf_3+ag_2)T^3+af_2T^2-dT-c$$ 
is
\begin{align}\label{eq_d}
	\small
\begin{pmatrix}
\ddots&\ddots&\ddots&\ddots&\ddots&\ddots&\\
&-c_{2i}&-d_{2i}&0&a_{2i}g_{2i+2}+b_{2i}f_{2i+3}&b_{2i}g_{2i+3}&b_{2i}h_{2i+3}&\\
&&-c_{2i+1}&-d_{2i+1}&a_{2i+1}f_{2i+3}&a_{2i+1}g_{2i+3}&a_{2i+1}h_{2i+3}+b_{2i+1}g_{2i+4}&0\\
&&&\ddots&\ddots&\ddots&\ddots&\ddots&\ddots
\end{pmatrix}.
\end{align}
Now we make some row transformations on the linear equations \eqref{eq_linear}.
Subtracting 
$
{b_{2i}h_{2i+3}}/{(a_{2i+1}h_{2i+3}+b_{2i+1}g_{2i+4})}$ times the $2i+1$ line from the $2i$ line to eliminate the term $b_{2i}h_{2i+3}$,
we can get a difference operator $\D$ with supports $\{0,1,2,3,4\}$ such that
\begin{align*}
\D\Phi=0.
\end{align*}
According to proposition \ref{pro_equi},
we can obtain a pair of twisted $N$-gons from the vector $\Phi$.
Taking another triple of difference operators
$ \alpha\D_1\beta^{-1},\quad
\alpha\D_2\gamma^{-1},\quad
\beta\D_3\gamma^{-1},$ 
a new solution is given by $\alpha\D\gamma^{-1}$,
and the corresponding pair of twisted $N$-gons are invariant.

On the other hand,
given two twisted $N$-gons $(v^1,v^2)\in\mathscr{P}_{3,N}$,
let $(V^1,V^2)$ be an arbitrary lift with the same monodromy $M\in\mathrm{SL}(\mathbb{R},4)$.
Define an infinite-dimensional vector $\Phi$ as in \eqref{def_phi}.
According to proposition \ref{pro_equi},
there is a unique $2N$-periodic difference $\D\in\PBDO_{2N}(\{0,1,2,3,4\})/H\tilde{\times}H$ such that
\begin{align*}
	\D\Phi=0.
\end{align*}
Suppose that $\D=x+yT+zT^2+wT^3+uT^4$,
and its matrix form is
\begin{align*}
\begin{pmatrix}
\ddots&\ddots&\ddots&\ddots&\ddots&\\
&x_{2i}&y_{2i}&z_{2i}&w_{2i}&u_{2i}&\\
&&x_{2i+1}&y_{2i+1}&z_{2i+1}&w_{2i+1}&u_{2i+1}&\\
&&&\ddots&\ddots&\ddots&\ddots&\ddots&\\
\end{pmatrix}.
\end{align*}
Subtracting ${z_{2i}}/{y_{2i+1}}$ times the $2i+1$ line from the $2i$ line to eliminate the term $z_{2i}$,
we can get a difference operator in the form \eqref{eq_d}
\begin{align}\label{eq_matrix4}
\begin{pmatrix}
	\ddots&\ddots&\ddots&\ddots&\ddots&\ddots&\\
	&*&*&0&\diamond&\diamond&\diamond&\\
	&&*&*&\diamond&\diamond&\diamond&0&\\
	&&&\ddots&\ddots&\ddots&\ddots&\ddots&\ddots&\\
\end{pmatrix}.
\end{align}
The elements in position $*$ constitute a difference operator with supports $\{0,1\}$,
and we denote it by $-\D_2$.
Besides, the elements in position $\diamond$ constitute a difference operator supported in $\{2,3,4,5\}$,
and we denote it by $\D_4$.
According Lemma \ref{lemma_fac},
there is a unique factorization $\D_4=\D_1\D_3$ up to the action \eqref{eq_13}.
Since $\D$ is unique up to the action
\begin{align*}
\D\mapsto
\alpha\D\gamma^{-1},\quad
(\alpha,\beta)\in H\tilde{\times}H,
\end{align*}
these three $2N$-periodic difference operators $\D_i$, $i=1,2,3$ are unique up to the action \eqref{eq_map}.
\end{proof}

\subsection{Refactorization}
According to Lemma \ref{lemma_d123},
we have
\begin{align}\label{eq_d0}
	\begin{split}
		\D_1\D_3{\Phi}&=\D_2\Phi,\\
		\tilde{\Phi}&=\D_3\Phi.
	\end{split}
\end{align}

\begin{lemma}\label{lemma_fa}

Suppose that $\D=fT^2+gT^3+hT^4+pT^5+qT^6$ and $\hat{\D}_2=a+bT$ are $2N$-periodic difference operators,
such that
$f_{2i}=q_{2i+1}=0$ and $b_i\neq 0,$ $\forall i\in\mathbb{Z}$.
If $\ker\hat{\D}_2\subset\ker \D$,
then the operator $\D$ admits a  factorization with the form
\begin{align*}
\D=\D_4\hat{\D}_2,
\end{align*}
where $\D_4=xT^2+yT^3+zT^4+wT^5$  such that $x_{2i}=w_{2i+1}=0$ for all $i\in\mathbb{Z}$.
\end{lemma}
\begin{proof}

Suppose that $\xi\in\ker\hat{\D}_2$,
that is
$
\hat{\D}_2\xi=0,$
which yields $a_i\xi_i+b_{i}\xi_{i+1}=0$, $\forall i\in\mathbb{Z}$.
Besides, the identity $\D\xi=0$ give us
\begin{align*}
g_{2i}\xi_{2i+3}+h_{2i}\xi_{2i+4}+p_{2i}\xi_{2i+5}+q_{2i}\xi_{2i+6}&=0,\\
f_{2i+1}\xi_{2i+3}+g_{2i+1}\xi_{2i+4}+h_{2i+1}\xi_{2i+5}+p_{2i+1}\xi_{2i+6}&=0,\quad\forall i\in\mathbb{Z}.
\end{align*}
Comparing the coefficients of $T^j$ in $\D=\D_4\hat{\D}_2$ and using the identities above,
we obtain the solution for $\D_4=xT^2+yT^3+zT^4+wT^5$ 
\begin{align*}
x_{2i}=0,\quad
x_{2i+1}=
\frac{a_{2i+4}a_{2i+5}p_{2i+1}+b_{2i+4}b_{2i+5}g_{2i+1}-a_{2i+4}b_{2i+5}h_{2i+1}}{b_{2i+3}b_{2i+4}b_{2i+5}},\\
y_{2i}=\frac{a_{2i+4}a_{2i+5}q_{2i}+b_{2i+4}b_{2i+5}h_{2i}-a_{2i+4}b_{2i+5}p_{2i}}{b_{2i+3}b_{2i+4}b_{2i+5}},\quad
y_{2i+1}=
\frac{b_{2i+5}h_{2i+1}-a_{2i+5}p_{2i+1}}{b_{2i+4}b_{2i+5}},\\
z_{2i}=
\frac{b_{2i+5}p_{2i}-a_{2i+5}q_{2i}}{b_{2i+4}b_{2i+5}},\quad
z_{2i+1}=
\frac{p_{2i+1}}{b_{2i+5}},\quad
w_{2i}=
\frac{q_{2i}}{b_{2i+5}},\quad
w_{2i+1}=0,\quad
\forall i\in\mathbb{Z}.
\end{align*}
\end{proof}

\begin{proposition}\label{pro_map}

The equation
\begin{align}\label{eq_ref}
\left(\tilde{\D}_1^{-1}\tilde{\D}_2\right)\D_3=\tilde{\D}_3\left(\D_1^{-1}\D_2\right)
\end{align}
gives a well-defined mapping  of the space of triples of operators satisfying the conditions in lemma \ref{lemma_d123}, up to the action \eqref{eq_map}, to itself: $\D_i\mapsto \tilde{\D}_i$, $i=1,2,3$.

\end{proposition}

\begin{proof}

Firstly,
we will construct two difference operators $\hat{\D}_i$, $i=1,2,$ such that
\begin{align}\label{eq_d12}
	\D_1\hat{\D}_2=\D_2\hat{\D}_1,
\end{align}
where $\hat{\D}_1$ is supported in $\{2,3\}$ and $\hat{\D}_2$ is supported in $\{0,1\}$.
Recall that the duality of a difference operator $\D=\sum a_iT^i$ is defined as
\begin{align*}
	\D^*=\sum T^{-i}a_i.
\end{align*} 
This corresponds to transposition of the operator matrix.
Consider the duality of \eqref{eq_d12}
\begin{align*}
\hat{\D}_2^*\D_1^*=\hat{\D}_1^*\D_2^*.
\end{align*}
This is a standard refactorization mapping,
so the existence and uniqueness of $\hat{\D}_i^*$, $i=1,2$, up to a left action $\alpha\hat{\D}_i^*$ for a quasi-periodic sequence $\alpha$,
can be obtained according to \cite[Lemma 4.10]{refactorization2020}.
Then taking duality of $\hat{\D}_i^*$, $i=1,2,$ 
we could get $\hat{\D}_i$, $i=1,2$.

Substituting $\D_1^{-1}\D_2=\hat{\D}_2\hat{\D}_1^{-1}$ into \eqref{eq_ref},
we get
\begin{align}\label{eq_re2}
\tilde{\D}_2\left(\D_3\hat{\D}_1\right)=\left(\tilde{\D}_1\tilde{\D}_3\right)\hat{\D}_2.
\end{align}
Here $\D_3\hat{\D}_1$, $\hat{\D}_2$ are known and we would like to solve $\tilde{\D}_i$, $i=1,2,3$.
Applying both sides of \eqref{eq_re2} to any $\xi\in\ker\hat{\D}_2$,
we get $\tilde{\D}_2\left(\D_3\hat{\D}_1\right)\xi=0$,
meaning that
\begin{align*}
\D_3\hat{\D}_1\left(\ker \hat{\D}_2\right)\subset\ker\tilde{\D}_2.
\end{align*}
Since the operators $\hat{\D}_2$ and $\hat{\D}_1$ have trivially intersecting kernels,
and the kernel of $\D_3$ is empty,
it follows that $$\dim\ker\tilde{\D}_2=\dim\D_3\hat{\D}_1\left(\ker \hat{\D}_2\right)=\dim\ker\hat{\D}_2=1.$$
Suppose that $\xi\in\mathbb{R}^\infty$ is a nontrivial element in $\ker\hat{\D}_2$,
then $\eta:=\D_3\hat{\D}_1\xi$ is an element in $\ker\tilde{\D}_2$.
We can construct $\tilde{\D}_2$ as
\begin{align*}
\tilde{\D}_2=\eta T-\eta_1,
\end{align*}
where $\eta_1$ is the sequence $T\eta$.
Here $\tilde{\D}_2$ is unique up to the left action $\alpha\tilde{\D}_2$ for a quasi-periodic sequence $\alpha$.

Next,
we solve $\tilde{\D}_1$ and $\tilde{\D}_3$.
Let $\D_4:=\tilde{\D}_1\tilde{\D}_3$ and $\D_5:=\tilde{\D}_2\left(\D_3\hat{\D}_1\right)$,
then the equation \eqref{eq_re2} can be written as
\begin{align*}
\D_4\hat{\D}_2=\D_5.
\end{align*}
Since $\D_4$ and $\D_5$ satisfy the conditions of Lemma \ref{lemma_fa},
we can get a solution for $\D_4$.
According to Lemma \ref{lemma_fac},
we factor $\D_4$ and get $\tilde{\D}_1$, $\tilde{\D}_3$.
Taking another triple of difference operators
$ \alpha\D_1\beta^{-1},$ $
\alpha\D_2\gamma^{-1},$ $
\beta\D_3\gamma^{-1},$ 
a new solution for $\hat{\D}_1$, $\hat{\D}_2$ is given by
\begin{align*}
\beta\hat{\D}_2\delta^{-1},\quad
\gamma\hat{\D}_1\delta^{-1},
\end{align*}
and the new solution for $\tilde{\D}_i$, $i=1,2,3$, is
\begin{align*}
\rho\tilde{\D}_1\epsilon^{-1},\quad
\rho\tilde{\D}_2\beta^{-1},\quad
\epsilon\tilde{\D}_3\beta^{-1},
\end{align*} 
for some quasi-periodic sequences $\alpha$, $\beta$, $\gamma$, $\epsilon,$ $\delta$, $\rho$ with the same monodromy.
\end{proof}

\begin{theorem}
The map \eqref{map1},
written in terms of difference operators,
is a refactorization map determined by \eqref{eq_ref}.

\end{theorem}
\begin{proof}
Firstly,
we construct $\Phi$ and $\tilde{\Phi}$ using \eqref{def_phi} and \eqref{def_phi2}.
Applying both sides of \eqref{eq_ref} to $\Phi$,
we get
\begin{align*}
\left(\tilde{\D}_1^{-1}\tilde{\D}_2\right)\D_3
\Phi
=\tilde{\D}_3\left(\D_1^{-1}\D_2\right)
\Phi,
\end{align*}
which,
using $\D_1\D_3\Phi=\D_2\Phi$ and thus $\D_3\Phi=\D_1^{-1}\D_2\Phi$,
can be rewritten as
\begin{align*}
\tilde{\D}_1\tilde{\D}_3\D_3
\Phi
=\tilde{\D}_2\D_3
\Phi
.
\end{align*}
This means that $\tilde{\Phi}=\D_3\Phi$,
which is just \eqref{eq_dd}.
\end{proof}

\begin{lemma}\label{lemma_bi}
Suppose that $\J_1=\{2,3\}$, $\J_2=\{0,1\}$ and $\J_3=\{0,1,2\}$.
Let $\SDO_{2N}(\J_3)$ stand for the set of difference operators $f+gT+hT^2\in\DO_{2N}(\J_3)$ with $f_{2i}=h_{2i+1}=0$, $\forall i\in\mathbb{Z}$.
The map 
\begin{align*}
\PBDO_{2N}(\J_1)\times\PBDO_{2N}(\J_2)\times\SDO_{2N}(\J_3)/\hhh\\
\rightarrow
\PBDO_{2N}(\J_2)^{-1}\PBDO_{2N}(\J_1)\SDO_{2N}(\J_3)/\Ad H
\end{align*}
given by
\begin{align*}
\left({\D}_1,\,{\D}_2,\,\D_3\right)
\longmapsto
{\D}_2^{-1}\D_1{\D}_3
\end{align*}
is generically a bijection,
where $\hhh$ is the group of triples of non-vanishing $2N$-quasi-periodic sequences with the same monodromy acting by the action \eqref{eq_map}.

\end{lemma}

\begin{proof}
It is clearly surjective by definition of the codomain, so it suffices to prove injectivity.
Suppose that ${\D}_2^{-1}\D_1{\D}_3$ is conjugate to $\bar{\D}_2^{-1}\bar{\D}_1\bar{\D}_3$,
that is
\begin{align*}
{\D}_2^{-1}\D_1{\D}_3=
\alpha^{-1}
\bar{\D}_2^{-1}\bar{\D}_1\bar{\D}_3
\alpha
=
\left(\bar{\D}_2\alpha
\right)^{-1}
\bar{\D}_1\bar{\D}_3
\alpha
\end{align*}
for some $2N$-periodic sequence $\alpha\in H$.
According to proposition \ref{pro_map},
there exist three difference operators $\tilde{\D}_i$, $i=1,2,3$, such that the following equation (similar to \eqref{eq_re2})
\begin{align*}
\tilde{\D}_2\left(\D_3^*{\D}_1^*\right)=\left(\tilde{\D}_1\tilde{\D}_3\right){\D}_2^*.
\end{align*}
Taking its dual,
we get
\begin{align*}
{\D}_2\left(\tilde{\D}_3^*\tilde{\D}_1^*\right)=\left(
{\D}_1{\D}_3\right)\tilde{\D}_2^*,
\end{align*}
which yields
\begin{align*}
{\D}_2^{-1}\D_1{\D}_3=
\left(\tilde{\D}_3^*\tilde{\D}_1^*\right){(\tilde{\D}_2^*)}^{-1}
=
\left(\bar{\D}_2\alpha
\right)^{-1}
\bar{\D}_1\bar{\D}_3
\alpha,
\end{align*}
that is
\begin{align*}
\left(\bar{\D}_2\alpha
\right)\left(\tilde{\D}_3^*\tilde{\D}_1^*\right)
=
\bar{\D}_1\bar{\D}_3
\alpha\tilde{\D}_2^*.
\end{align*}
According to proposition \ref{pro_map},
$\D_i$ and $\bar{\D}_i$, $i=1,2,3$, are in the same orbit of \eqref{eq_map}.
\end{proof}

\begin{proposition}

The map in proposition \ref{pro_map} is a Poisson map.

\end{proposition}

\begin{proof}

The map determined by equation \eqref{eq_ref} can be treated as the composition of the following two maps 
\begin{align*}
\left(\D_1,\,\D_2,\,\D_3\right)\longrightarrow
\left(\hat{\D}_1,\,\hat{\D}_2,\,{\D}_3\right)\longrightarrow
\left(\tilde{\D}_1,\,\tilde{\D}_2,\,\tilde{\D}_3\right).
\end{align*}
The first arrow is determined by \eqref{eq_d12}.
This is the inverse of the refactorization map \cite{refactorization2020},
so
it is Poisson.
The second arrow can be treated as the following composition
\begin{align*}
\left(\hat{\D}_1,\,\hat{\D}_2,\,{\D}_3\right)\longrightarrow
\D_3\hat{\D}_1\hat{\D}_2^{-1}=\tilde{\D}_2^{-1}\tilde{\D}_1\tilde{\D}_3
\longrightarrow
\left(\tilde{\D}_1,\,\tilde{\D}_2,\,\tilde{\D}_3\right).
\end{align*}
The first arrow is a map 
\begin{align*}
\PBDO_{2N}(\J_1)\times\PBDO_{2N}(\J_2)\times\SDO_{2N}(\J_3)/\hhh
\\
\rightarrow
\PBDO_{2N}(\J_2)^{-1}\PBDO_{2N}(\J_1)\SDO_{2N}(\J_3)/\Ad H.
\end{align*}
According to proposition \ref{pro_map},
this map is well-defined.
Besides,
it is Poisson,
since multiplication in the group of pseudo-difference operators is Poisson,
inversion is anti-Poisson.
By lemma \ref{lemma_bi},
the second arrow is generically invertible.
Therefore, the whole map is a composition of Poisson maps and hence Poisson.
\end{proof}
The equation \eqref{eq_ref} is invariant under the transformation
\begin{align*}
\D_2\mapsto z\D_2,\quad
\tilde{\D}_2\mapsto z\tilde{\D}_2,
\end{align*} 
under which,
\eqref{eq_d0} becomes
\begin{align}
	\begin{split}
		\D_1\D_3{\Phi}&=z\D_2\Phi,\\
		\tilde{\Phi}&=\D_3\Phi.
	\end{split}
\end{align}
Define
\begin{align*}
\mathcal{L}:=\D_3^{-1}\D_1^{-1}\D_2
\end{align*}
the map has the Lax representation
\begin{align*}
\mathcal{L}\mapsto
\tilde{\mathcal{L}}:=\D_1^{-1}\D_2\D_3^{-1}.
\end{align*}
Inspired by that,
we can insert  a spectral parameter $z$ in the linear problem \eqref{eq_lax}
\begin{align*}
	L_i^j:=\begin{pmatrix}
		0&0&1&0\\
		0&0&0&1\\
		za_i^j&zb_i^j&z&c_i^j\\
		zf_i^j&zg_i^j&z&h_i^j
	\end{pmatrix},\quad
	M_i^j:=
	\begin{pmatrix}
		0&1&0&0\\
		0&\frac{\alpha_{i-1}^j}{\lambda_{i-1}^j}&\frac{f_{i-1}^j}{\lambda_{i-1}^j}&-\frac{f_{i-1}^j}{\lambda_{i-1}^j}\\
		0&0&0&1\\
		z\frac{\beta_i^j}{\lambda_i^j}&z\frac{\gamma_i^j}{\lambda_i^j}&0&\frac{\eta_i^j}{\lambda_i^j}
	\end{pmatrix}.
\end{align*}
\subsection{Poisson brackets}

Here we derive explicit formulas for Poisson brackets preserved by the map $T_Q$.
The phase space of this map is the space $\mathscr{P}_{3,N}$ of pairs of twisted $N$-gons $(v_i^1,v_i^2)_{i\in\mathbb{Z}}$ in $\mathbb{RP}^3$.
The space can be coordinated using six $N$-periodic sequences $a_i$, $b_i$, $c_i$, $f_i$, $g_i$, $h_i$ defined using cross-ratios \eqref{eq_cr}.
\begin{proposition}
In these coordinates $\{a_i,\,b_i,\,c_i,\,f_i,\,g_i,\,h_i\}$,
the Poisson structure of the map $T_Q$ takes the form \eqref{eq_Poi}.
\end{proposition}
\begin{proof}
According to proposition \ref{pro_gon},
there are three difference operators $\D_i$,
$i=1,2,3$ as in lemma \ref{lemma_d123},
corresponding to a pair of twisted $N$-gons $(v^1,v^2)\in\mathscr{P}_{3,N}$.
Suppose that
\begin{align*}
	\D_1\D_3=\alpha T^2+\beta T^3 +\gamma T^4+\delta T^5,\quad
	\D_2=\xi +\eta T,
\end{align*}
where $\alpha$, $\beta$, $\gamma$, $\delta$, $\xi$ and $\eta$ are $2N$-periodic sequences,
such that
\begin{align}\label{eq_redu}
	\begin{split}
		\alpha_{2i}=0,\quad
		\beta_{2i}=C_i,\quad
		\gamma_{2i}=D_i,\quad
		\delta_{2i}=R_i,\\
		\alpha_{2i+1}=G_i,\quad
		\beta_{2i+1}=H_i,\quad
		\gamma_{2i+1}=S_i,\quad
		\delta_{2i+1}=0,\\
		\xi_{2i}=A_i,\quad
		\xi_{2i+1}=E_i,\quad
		\eta_{2i}=B_i,\quad
		\eta_{2i+1}=F_i,
	\end{split}
\end{align}
for all $i\in\mathbb{Z}$.
Assume that $(\phi^1,\phi^2)$ is an arbitrary lift of the pair of twisted $N$-gons.
Then we can construct a sequence $\Phi$ using \eqref{def_phi}, such that $\D_1\D_3\Phi=z\D_2\Phi$.
Equivalently,
we have
\begin{align*}
	z\left(A_i\phi_{i-1}^1+B_i\phi_{i-1}^2\right)
	&
	=C_i\phi_i^2+D_i\phi_{i+1}^1+R_i\phi_{i+1}^2,\\
	z\left(E_i\phi_{i-1}^2+F_i\phi_{i}^1\right)
	&
	=G_i\phi_i^2+H_i\phi_{i+1}^1+S_i\phi_{i+1}^2,
\end{align*}
where $z$ is the spectral parameter.
Using the cross-ratio expressions \eqref{eq_cr},
the coordinates $a_i$, $b_i$, $c_i$, $f_i$, $g_i$, $h_i$ can be expressed in terms of coefficients of $\D_1\D_3$ and $\D_2$ as follows:
\begin{align}\label{eq_aA}
	\begin{split}
		a_i=\frac{A_iS_i\theta_{i-1}}{F_iF_{i-1}R_iR_{i-1}},\quad
		b_i=-\frac{D_{i-2}\theta_{i-1}\rho_i}{F_iF_{i-1}R_iR_{i-1}R_{i-2}},\quad
		c_i=-\frac{D_{i-1}\zeta_{i}}{F_iR_iR_{i-1}},\\
		f_i=\frac{A_iH_i\theta_{i-1}}{D_iF_iF_{i-1}R_{i-1}},\quad
		g_i=-\frac{D_{i-2}\xi_i\theta_{i-1}}{D_iF_iF_{i-1}R_{i-1}R_{i-2}},\quad
		h_i=-\frac{D_{i-1}\eta_{i}}{D_{i}F_iR_{i-1}},
	\end{split}
\end{align}
where 
\begin{align*}
	\theta_i=D_{i}S_{i}-H_{i}R_{i},\quad
	\rho_i=B_{i}S_{i}-E_{i}R_{i},\quad
	\zeta_i=C_{i}S_{i}-G_{i}R_{i},\\
	\xi_i=B_{i}H_{i}-E_{i}D_{i},\quad
	\eta_i=C_{i}H_{i}-G_{i}D_{i}.
\end{align*} 
The corresponding matrix of $\D_2$ is
\begin{align*}
	\begin{matrix}
		\ddots&\ddots&&&&&\\
		&A_i&B_i&&&&\\
		&&E_i&F_i&&&\\
		&&&A_{i+1}&B_{i+1}&&&&\\
		&&&&E_{i+1}&F_{i+1}&&&\\
		&&&&&\ddots&\ddots&&\\
	\end{matrix}
\end{align*}
while the corresponding matrix of $\D_1\D_3$ is
\begin{align*}
	\begin{matrix}
		\ddots&\ddots&\ddots&\ddots&&&&&&\\
		&0&C_i&D_i&R_i&&&&&\\
		&&G_i&H_i&S_i&0&&&&\\
		&&&0&C_{i+1}&D_{i+1}&R_{i+1}&&&\\
		&&&&G_{i+1}&H_{i+1}&S_{i+1}&0&&\\
		&&&&&\ddots&\ddots&\ddots&\ddots&\\
	\end{matrix}
\end{align*}
The corresponding Poisson brackets are defined as
\begin{align}\label{eq_Po2}
	\begin{split}
		\left[A_i,B_i \right]=\frac{1}{2}A_iB_i,\quad
		\left[E_i,F_i \right]=\frac{1}{2}E_iF_i,\quad
		\left[B_i,E_i \right]=\frac{1}{2}B_iE_i,\quad
		\left[F_i,A_{i+1} \right]=\frac{1}{2}F_iA_{i+1}.
	\end{split}
\end{align}
and
\begin{align}\label{eq_Po1}
	\begin{split}
		\left[D_i,C_i \right]=\frac{1}{2}D_iC_i,\quad
		\left[R_i,C_i \right]=\frac{1}{2}R_iC_i,\quad
		\left[R_i,D_i \right]=\frac{1}{2}R_iD_i,\quad
		\left[S_i,H_i \right]=\frac{1}{2}S_iH_i,\\
		\left[H_i,G_i \right]=\frac{1}{2}H_iG_i,\quad
		\left[S_i,G_i \right]=\frac{1}{2}S_iG_i,\quad
		\left[G_i,C_i \right]=\frac{1}{2}G_iC_i,\quad
		\left[H_i,D_i \right]=\frac{1}{2}H_iD_i,\\
		\left[S_i,R_i \right]=\frac{1}{2}S_iR_i,\,
		\left[C_{i+1},R_i \right]=\frac{1}{2}C_{i+1}R_i,\,
		\left[C_{i+1},Q_i \right]=\frac{1}{2}C_{i+1}S_i,\,
		\left[G_{i+1} ,R_i\right]=\frac{1}{2}G_{i+1}R_i,\\
		\left[G_{i+1},Q_i \right]=\frac{1}{2}G_{i+1}S_i,\quad
		\left[H_i,C_i \right]=D_iG_i,\quad
		\left[S_i,C_i \right]=G_iR_i,\quad
		\left[S_i,D_i \right]=H_iR_i.
	\end{split}
\end{align}
Then
the Poisson brackets \eqref{eq_Poi} can be computed using \eqref{eq_Po2}, \eqref{eq_Po1} and \eqref{eq_aA} through a straightforward calculation.

It can be verified that the map \eqref{eq_sys} preserves these brackets \eqref{eq_Poi} using a computer algebra system.
\end{proof}

\subsection{Other evolutions}\label{subsec_other}
According to Lemma \ref{lemma_d123},
we have
\begin{align}\label{e1}
	\begin{split}
		\D_1\D_3\Phi&=\D_2\Phi,\\
		\tilde{\Phi}&=\D_3\Phi,
	\end{split}
\end{align}
where the second equation represent the discrete evolution of $\Phi$. 
It is natural to ask whether there are other possible evolutions.
For example, consider 
\begin{align}\label{e2}
	\bar{\Phi}=\D_2\Phi=\D_1\D_3\Phi.
\end{align}
Since $\D_3\in\SDO_{2N}$,
we have $\D_1\D_3$ is a difference operator supported in $\{2,3,4,5\}$ with the special form in Lemma \ref{lemma_fac}.
Thus, the map is a reduction of the standard refactorization mapping with $\J_+=\{0,1\}$ and $\J_-=\{2,3,4,5\}$.
The two different evolutions \eqref{e1} and \eqref{e2} both preserve the Poisson structure \eqref{eq_Poi}.
However,
the geometric meaning of \eqref{e2} is unknown.

Recall that the inverse dented pentagram map $T_{I,J}$ in $\mathbb{RP}^4$ with $I=(1,1,1)$ and $J=(1,1,2)$ is defined as \cite{dented}
\begin{align*}
	T_{I,J}v_i=&\left(v_i,v_{i+1},v_{i+2},v_{i+3}\right)
	\cap
	\left(v_{i+1},v_{i+2},v_{i+3},v_{i+4}\right)\\
	&\cap
	\left(v_{i+2},v_{i+3},v_{i+4},v_{i+5}\right)
	\cap
	\left(v_{i+4},v_{i+5},v_{i+6},v_{i+7}\right)\\
	=&\left(v_{i+2},v_{i+3}\right)
	\cap
	\left(v_{i+4},v_{i+5},v_{i+6},v_{i+7}\right),
\end{align*}
for $i\in\mathbb{Z}$.
It is related to the standard refactorization mapping with $\J_+=\{0,1\}$ and $\J_-=\{2,3,4,5\}$.
It is open whether there are connections between our map \eqref{map1} and the inverse dented pentagram map above.

\section{Reduction to $B$-net}\label{sec_redu}

The B-net (or B-quadrilateral lattice) was firstly introduced by Doliwa \cite{Doliwa_BQL}.
In this section,
we consider the reduction of reduced $Q$-net to reduced B-net.

\begin{define}\label{def_B}
	A reduced $Q$-net $v:\mathbb{Z}^2\rightarrow \mathbb{RP}^N$ of type $Q=\{a,b,c\}$ is called a reduced $B$-net if the points $v_r$, $v_{r+a+b}$, $v_{r+a+c}$ and $v_{r+b+c}$ are coplanar.
\end{define}

In a $B$-net,
the points $v_{r+a}$, $v_{r+b}$, $v_{r+c}$ and $v_{r+a+b+c}$ are also coplanar.
According to \cite{Doliwa_BQL},
under the hypotheses of Definition \ref{def_B},
there exist a gauge such that 
\begin{align}\label{eq_vv}
	\bar{V}_{r+p+q}
	=\bar{V}_r+\gamma^{pq}_r\left(\bar{V}_{r+p}-\bar{V}_{r+q}\right),
\end{align}
where $p,q\in\{a,b,c\}$ are distinct and the coefficients $\gamma^{pq}_r$ depend on the positions of the points $\bar{V}_{r}\in\mathbb{R}^4$.
The equations above are just reductions of the linear problem of the well-known discrete BKP  \cite{Nimmo1997}.
The compatibility conditions of \eqref{eq_vv} lead to the following set of nonlinear equations:
\begin{align}
	1+\gamma^{qs}_{r+p}\left(\gamma^{pq}_{r}-\gamma^{ps}_r\right)=\gamma^{ps}_{r+q}\gamma^{pq}_r=\gamma^{pq}_{r+s}\gamma^{ps}_r,
\end{align}
where $p,q,s\in\{a,b,c\}$ are distinct with $\gamma^{pq}=-\gamma^{qp}$.
The second equality in the compatibility condition implies the existence of the potential $\tau:\mathbb{Z}^2\rightarrow \mathbb{R}$,
in terms of which the function $\gamma^{ij}_r$ can be written as
\begin{align*}
	\gamma^{pq}_r=\frac{\tau_{r+p}\tau_{r+q}}{\tau_r\tau_{r+p+q}},
\end{align*}
where $p,q\in\{a,b,c\}$ and $p\neq q$.
Then the first equality can be rewritten as the following reduction of the discrete BKP equations
\begin{align*}
	\tau_r\tau_{r+a+b+c}=\tau_{r+a+b}\tau_{r+c}-\tau_{r+a+c}\tau_{r+b}+\tau_{r+b+c}\tau_{r+a}.
\end{align*}

\begin{example}
For $Q=\{(-1,1),(0,1),(1,1)\}$,
the equation \eqref{eq_vv} is
\begin{align}\label{eq_acg}
\begin{split}
\bar{V}_{i-1}^{j+2}-\bar{V}_i^j
&=A_i^j(\bar{V}_{i-1}^{j+1}-\bar{V}_i^{j+1}),\\
\bar{V}_{i+1}^{j+2}-\bar{V}_i^j
&=G_i^j(\bar{V}_{i}^{j+1}-\bar{V}_{i+1}^{j+1}),\\
\bar{V}_{i}^{j+2}-\bar{V}_i^j
&=C_i^j(\bar{V}_{i-1}^{j+1}-\bar{V}_{i+1}^{j+1}).
\end{split}
\end{align}
Its compatibility conditions yield the following discrete  system 
\begin{align}\label{eq_et}
A_i^{j+1}=
\frac{A_{i-1}^j}{\eta_{i-1}^j},\quad
G_i^{j+1}=
\frac{G_{i+1}^j}{\eta_{i+1}^j},\quad
C_i^{j+1}=
\frac{C_{i}^j}{\eta_{i}^j},
\end{align}
where 
\begin{align}\label{eq_eta}
\eta_i^j:=A_i^jC_i^j+C_i^jG_i^j-A_i^jG_i^j.
\end{align}
Now we introduce the $\tau$-function
\begin{align*}
A_i^j=\frac{z_1+z_2}{z_1-z_2}\frac{\tau_{i-1}^{j+1}\tau_i^{j+1}}{\tau_i^j\tau_{i-1}^{j+2}},\quad
G_i^j=\frac{z_2+z_3}{z_2-z_3}\frac{\tau_{i+1}^{j+1}\tau_i^{j+1}}{\tau_i^j\tau_{i+1}^{j+2}},\quad
C_i^j=\frac{z_3+z_1}{z_1-z_3}\frac{\tau_{i-1}^{j+1}\tau_{i+1}^{j+1}}{\tau_i^j\tau_{i}^{j+2}},
\end{align*}
where $z_i$, $i=1,2,3$ are constants.
Under this transformation,
the system \eqref{eq_et} can be bilinearized
\begin{align}\label{eq_tau}
y_1\tau_i^j\tau_{i}^{j+3}+y_2\tau_i^{j+1}\tau_i^{j+2}+y_3\tau_{i-1}^{j+1}\tau_{i+1}^{j+2}+y_4\tau_{i-1}^{j+2}\tau_{i+1}^{j+1}=0,
\end{align}
where
\begin{align*}
y_1=(z_1-z_2)(z_2-z_3)(z_3-z_1),\quad
y_2=(z_1+z_2)(z_2+z_3)(z_3-z_1),\\
y_3=(z_1+z_2)(z_3+z_1)(z_2-z_3),\quad
y_4=(z_2+z_3)(z_3+z_1)(z_1-z_2).
\end{align*}
Consider the transformation
\begin{align*}
(n,t)=(i,-\delta j),
\end{align*}
and the reduction
\begin{align*}
\frac{y_2}{y_1}=2\delta^2-1,\quad
\frac{y_3}{y_1}=\frac{y_4}{y_1}=-\delta^2,
\end{align*}
then the equation \eqref{eq_tau} becomes
\begin{align*}
\tau_n^t\tau_n^{t+3}+(2\delta^2-1)\tau_n^{t+1}\tau_n^{t+2}-\delta^2\tau_{n-1}^{t+1}\tau_{n+1}^{t+2}
-\delta^2\tau_{n-1}^{t+2}\tau_{n+1}^{t+1}=0.
\end{align*}
This is equivalent to the discrete-time Toda equation of BKP type
\begin{align*}
	\left[e^{\frac{3}{2}\delta D_t}-e^{-\frac{1}{2}\delta D_t}-\delta^2\left(
	e^{D_n-\frac{1}{2}\delta D_t}+e^{-D_n-\frac{1}{2}\delta D_t}-2e^{-\frac{1}{2}\delta D_t}
	\right)
	\right]
	\tau \cdot \tau=0,
\end{align*}
which was derived in
\cite{hirota_iwao_tsujimoto_2001}.
Here $D_t$, $D_n$ are known Hirota operators.
Thus,
we provide a geometric explanation for this equation.
\end{example}

\begin{remark}\label{re_cr}
Using the cross-ratio expressions \eqref{eq_cr},
the coordinates $a_i$, $b_i$, $c_i$, $f_i$, $g_i$, $h_i$ for the pair of polygons $(\bar{V}^1,\bar{V}^2)$ in \eqref{eq_acg} can be expressed in terms of $A$, $C$, $G$
\begin{align*}
a_i=\frac{C_{i-1}\left(C_i-A_i\right)}{A_iA_{i-1}},\quad
b_i=\frac{C_{i-1}\eta_{i-2}}{A_iA_{i-1}A_{i-2}},\quad
c_i=-\frac{\eta_{i-1}}{A_iA_{i-1}},\\
f_i=\frac{C_{i-1}}{A_{i-1}},\quad
g_i=\frac{C_{i-1}\left(C_i-G_i\right)\eta_{i-2}}{A_{i-1}A_{i-2}\eta_i},\quad
h_i=\frac{G_i\eta_{i-1}}{A_{i-1}\eta_i},
\end{align*}
where $\eta_i$ is defined in \eqref{eq_eta}.
Omitting the variables $A$, $C$, $G$ above,
we obtain three equations
\begin{align}\label{eq_con}
a_i=f_i\left(1+f_{i+1}\right),\quad
g_i=\frac{b_i}{c_ic_{i+1}}\cdot\frac{b_{i+1}f_{i+1}+c_ic_{i+1}}{1+f_{i+1}},\quad
h_i=\frac{c_ic_{i+1}-b_{i+1}}{c_{i+1}\left(1+f_{i+1}\right)}.
\end{align}
Through a straightforward calculation,
the Poisson brackets \eqref{eq_Poi} are not preserved under the reduction determined by \eqref{eq_con}.
Therefore the reduction to the discrete BKP-type equation is not  Poisson.
\end{remark}

\section{Discussion}
In this paper,
we introduce a new family of generalizations of the pentagram map related to $Q$-nets.
The map can be treated as a
refactorization mapping related to the pseudo-difference operator $\D_1^{-1}\D_2$, such that $\D_1$ is supported in $\{2,3\}$ and $\D_2$ is supported in $\{0,1\}$.
There are some interesting open questions.

Firstly,
it is natural to consider the refactorization mapping with 
a general form
\begin{align*}
\mathcal{L}:=\left(\D_3^{-1}\D_4\right)\left(\D_1^{-1}\D_2\right)\mapsto
\tilde{\mathcal{L}}:=\left(\D_1^{-1}\D_2\right)\left(\D_3^{-1}\D_4\right),
\end{align*}
where $\D_i$, $i=1,2,3,4$ are periodic difference operators with different supports.
What is a geometric interpretation in this case?

According to subsection \ref{subsec_other},
the geometric explanation for evolution \eqref{e2} remains unknown.
Besides,
a question is whether there are connections between $Q$-pentagram map and the inverse dented pentagram map.

As noted in Remark \ref{re_cr},
the Poisson bracket \eqref{eq_Poi} cannot be reduced directly to the BKP case.
How to construct Poisson structure for \eqref{eq_et}?

It is known the pentagram map can be described as sequences of cluster mutations \cite{GLICK20111019,cluster2024}.
It would be interesting to find a similar description for $Q$-type pentagram maps.

The  continuous limit of the classical pentagram map is the Boussinesq equation \cite{2010The},
and the continuous limits of some generalized pentagram maps were considered in \cite{dented,wang2023pentagram}.
It is worth studying the continuous limits of $Q$-type pentagram maps.

\section*{Acknowledgement}
The author would like to thank Anton Izosimov
for his useful comments.
The author was supported by  the National Natural Science Foundation of China (Grant Nos. 12201325 and 12235007).

\begin{appendix}

\section{Poisson brackets of the map $T_Q$}
\begin{align}\label{eq_Poi}
	\begin{split}
		[a_i,b_i]=a_ib_i,\quad
		[c_i,a_i]=c_ia_i,\quad
		[h_i,a_i]=h_if_i,\quad
		[c_i,b_i]=c_ib_i,\quad
		[f_i,b_i]=f_ig_i,\\
		[h_i,b_i]=g_ih_i,\quad
		[f_{i},g_i]=f_ig_i,\quad
		[h_{i},f_i]=h_if_i,\quad
		[h_{i},g_i]=h_ig_i,\\
		[a_{i},a_{i-1}]=a_{i-1}a_i,\quad
		[b_{i},a_{i-1}]=a_{i-1}b_i,\quad
		[f_{i},a_{i-1}]=a_{i-1}f_i,\quad
		[g_{i},a_{i-1}]=a_{i-1}g_i,\\
		[a_{i+1},b_i]=a_{i+1}b_i,\quad
		[a_{i+1},c_i]=a_{i+1}c_i,\quad
		[a_{i+1},f_i]=a_{i+1}f_i,\quad
		[a_{i+1},g_i]=a_{i+1}g_i,\\
		[a_{i+1},h_i]=a_{i+1}h_i,\quad
		[b_i,b_{i-1}]=b_ib_{i-1},\quad
		[f_i,b_{i-1}]=f_ib_{i-1},\quad
		[g_i,b_{i-1}]=g_ib_{i-1},\\
		[b_{i+1},c_i]=b_{i+1}c_i,\quad
		[b_{i+1},f_i]=b_{i+1}f_i,\quad
		[b_{i+1},g_i]=b_{i+1}g_i,\quad
		[b_{i+1},h_i]=b_{i+1}h_i,\\
		[f_{i},f_{i-1}]=f_if_{i-1},\quad
		[g_{i},f_{i-1}]=g_if_{i-1},\quad
		[f_{i+1},c_{i}]=f_{i+1}c_{i},\quad
		[f_{i+1},g_{i}]=f_{i+1}g_{i},\\
		[f_{i+1},h_{i}]=f_{i+1}h_{i},\quad
		[g_{i},g_{i-1}]=g_ig_{i-1},\quad
		[g_{i+1},c_i]=g_{i+1}c_i,\quad
		[g_{i+1},h_i]=g_{i+1}h_i,\\
		[b_i,g_i]=g_i\left(b_i-g_i\right),\quad
		[a_{i-2},g_i]=g_i\left(f_{i-2}-a_{i-2}\right),\quad
		[a_{i-2},b_i]=b_i\left(f_{i-2}-a_{i-2}\right),\\
		[a_i,f_i]=f_i\left(a_i-f_i\right),\quad
		[c_{i},a_{i-1}]=c_i\left(a_{i-1}-f_{i-1}\right),\quad
		[h_{i},a_{i-1}]=h_i\left(a_{i-1}-f_{i-1}\right),\\
		[h_i,c_i]=h_i\left(h_i-c_i\right),\quad
		[b_{i-2},g_i]=g_i\left(g_{i-2}-b_{i-2}\right),\quad
		[c_i,b_{i-1}]=c_i\left(b_{i-1}-g_{i-1}\right),\\
		[h_i,b_{i-1}]=h_i\left(b_{i-1}-g_{i-1}\right),\quad
		[b_{i+2},c_i]=b_{i+2}\left(c_i-h_i\right),\quad
		[c_{i},c_{i-1}]=c_i\left(c_{i-1}-h_{i-1}\right),\\
		[f_{i},c_i]=h_i\left(f_i-a_i\right)-f_ic_i,\quad
		[g_{i+2},c_i]=g_{i+2}\left(c_i-h_i\right),\quad
		[b_{i-2},b_i]=b_i\left(g_{i-2}-b_{i-2}\right),\\
		[h_{i+1},c_i]=h_{i+1}\left(c_i-h_i\right),\quad
		[g_{i},c_i]=h_i\left(g_i-b_i\right)-g_ic_i,\quad
		[a_i,g_i]=a_ig_i+f_i\left(b_i-g_i\right).
	\end{split}
\end{align}

\end{appendix}


\end{document}